\documentclass[11pt]{article}
\usepackage{color}
\usepackage{amssymb,amsmath,amsfonts,amsthm,amsthm,mathtools}
\usepackage[english]{babel}
\usepackage[utf8]{inputenc}
\usepackage[margin=24mm]{geometry}
\usepackage{graphicx,color,tocloft}
\usepackage{times,fourier}
\definecolor{darkred}{rgb}{0.8,0,0}
\definecolor{darkblue}{rgb}{0,0,0.7}
\usepackage[colorlinks=true,citecolor=blue,linkcolor=darkblue]{hyperref}

\newlength\figwidth
\makeatletter
\renewcommand\section{\@startsection {section}{1}{\z@}%
  {-2.2ex \@plus -1ex \@minus -.2ex}%
  {1ex \@plus.1ex}%
  {\normalfont\bf\sffamily\color{darkblue}}}
\renewcommand\subsection{\@startsection{subsection}{2}{\z@}%
  {-1.85ex\@plus -0.4ex \@minus -.2ex}%
  {0.6ex \@plus .1ex}%
  {\normalfont\small\bf\sffamily}}
\renewcommand\subsubsection{\@startsection{subsubsection}{3}{\z@}%
  {-0.6ex\@plus -0.2ex \@minus -.2ex}%
  {0.4ex \@plus .1ex}%
  {\normalfont\normalsize\it}}
\renewcommand\paragraph{\@startsection{paragraph}{4}{\z@}%
  {0.2ex \@plus0.2ex \@minus0.1ex}{-0.5em}%
  {\normalfont\normalsize\bfseries}}
\def\ps@headings{%
  \let\@oddfoot\@empty
  \let\@evenfoot\@empty
  \def\@evenhead{\small\sffamily\thepage\hfil\slshape\leftmark}%
  \def\@oddhead{\small\sffamily{\slshape\rightmark}\hfil\thepage}%
  \let\@mkboth\markboth
  \def\chaptermark##1{\markboth{{\ifnum \c@secnumdepth >\m@ne
		\if@mainmatter \@chapapp\ \thechapter. \ \fi \fi ##1}}{}}%
  \def\sectionmark##1{\markright {{\ifnum \c@secnumdepth >\z@
		\thesection. \ \fi ##1}}}}
\def\fbf#1{\setbox0=\hbox{$#1$}\kern-0.10\wd0
  \lower0.02em\copy0\kern-\wd0 \lower0.02em\hbox{\kern+0.04em\copy0}\kern-\wd0
  \raise0.00em\copy0\kern-\wd0 \raise0.00em\hbox{\kern-0.04em\box0}}
\def\overl@ss#1#2{\vcenter{\offinterlineskip
        \ialign{$\m@th#1\hfil##\hfil$\crcr#2\crcr<\crcr } }}
\def\gl{\mathrel{\mathpalette\overl@ss>}}
\makeatother

\advance\textwidth -10mm
\advance\textheight -16mm
\advance\hoffset 4mm
\advance\voffset 4mm

\advance\parskip 4pt
\numberwithin{equation}{section}
2

\newtheorem{theorem}{Theorem}[section]
\newtheorem{definition}[theorem]{Definition}
\newtheorem{corollary}[theorem]{Corollary}
\newtheorem{lemma}[theorem]{Lemma}
\newtheorem{remark}[theorem]{Remark}
\newtheorem{proposition}[theorem]{Proposition}
\newtheorem{assumption}[theorem]{Assumption}

\def\maketitle{\par\noindent{\LARGE\bf\sffamily\color{red}\thetitle}\\[2ex]
{\large\theauthor}\\[1ex]
\textit{\thetextinfo}\\[1ex]
{\small\today}\par\vglue1.4\bigskipamount}
\def\title#1{\def\thetitle{#1}}
\def\author#1{\def\theauthor{#1}}
\def\textinfo#1{\def\thetextinfo{#1}}

\def\be{\begin{equation}}
\def\ee{\end{equation}}
\def\bse{\begin{subequations}}
\def\ese{\end{subequations}}
\newtheorem{RHP}{RHP}[section]

\definecolor{deeppurple}{rgb}{0.5, 0, 0.7}

\def\half{{\textstyle\frac12}}

\def\diag{\mathop{\rm diag}\nolimits}

\def\cn{\mathop{\rm cn}\nolimits}
\def\sn{\mathop{\rm sn}\nolimits}
\def\dn{\mathop{\rm dn}\nolimits}
\def\Natural{\mathbb{N}}
\def\Real{\mathbb{R}}
\def\R{\mathbb{R}}
\def\Complex{\mathbb{C}}

\def\Integer{\mathbb{Z}}
\def\Z{\mathbb{Z}}
\def\i{\text{i}}

\def\Re{\mathop{\rm Re}\nolimits}
\def\Im{\mathop{\rm Im}\nolimits}

\def\tr{\mathop{\rm tr}\nolimits}
\def\re{\mathrm{re}}
\def\im{\mathrm{im}}
\def\d{\mathrm{d}}
\def\sgn{\mathop{\rm sgn}\nolimits}
\def\e{\mathop{\rm e}\nolimits}
\def\@#1{{\mathbf{#1}}}
\def\_#1{{\mathsf{#1}}}
\let\~=\tilde
\let\==\bar
\let\^=\hat
\let\<=\langle
\let\>=\rangle

\def\L{\mathcal{L}}

\def\Dir{{\mathrm{Dir}}}
\def\note[#1]{\marginpar{\color{blue}[#1]}}
\def\partialderiv#1#2{{\frac{\partial #1}{\partial #2}}}
\def\txtfrac#1#2{{\textstyle\frac{#1}{#2}}}

\def\C{{\mathbb C}}
\def\R{{\mathbb R}}

\def\Z{{\mathbb Z}}

\def\D{\Delta}
\def\1{{\bf 1}}

\def\z{\zeta}
\def\l{\lambda}

\def\e{\mathrm{e}}
\def\g{g}
\def\mucirc{\mathring\mu}
\makeatother

\let\trueparagraph=\paragraph
\def\paragraph#1{\par\smallskip\trueparagraph{\rm\textbf{#1}}}
\advance\textheight 12mm
\advance\voffset -6mm
\allowdisplaybreaks

\def\em{\endgroup}

\begin{document}
\pagestyle{plain}

\title{Spectral theory for non-self-adjoint Dirac operators with\\[0.1ex]
periodic potentials and inverse scattering transform\\[0.1ex]
for the focusing nonlinear Schr\"odinger equation\\[0.1ex]
with periodic boundary conditions}
\author{Gino Biondini$^1$, Gregor Kova\v{c}i\v{c}$^2$, Alexander Tovbis$^3$, 
Zachery Wolski$^2$ and Zechuan Zhang$^4$}
\textinfo{$^1$ Department of Mathematics, State University of New York at Buffalo, Buffalo, New York 14260 - USA\\
$^2$ Department of Mathematical Sciences, Rensselaer Polytechnic Institute, Troy, New York 12180 - USA\\
$^3$ Department of Mathematics, University of Central Florida, Orlando, Florida 32816 - USA\\
$^4$ Scuola Internazionale Superiore di Studi Avanzati, Trieste - Italy}
\maketitle

\begin{abstract}
\noindent
We formulate the inverse spectral theory for a non-self-adjoint one-dimensional Dirac operator associated
periodic potentials via a Riemann-Hilbert problem approach.
We use the resulting formalism to solve the initial value problem for the focusing nonlinear Schr\"o\-dinger equation.
We establish a uniqueness theorem for the solutions of the Riemann-Hilbert problem, which provides a new method for obtaining the potential from the spectral data.
The formalism applies for both finite- and infinite-genus potentials.
As in the defocusing case, the formalism shows that only a single set of Dirichlet eigenvalues is needed in order to uniquely reconstruct 
the potential of the Dirac operator and the corresponding solution of the focusing NLS equation.
\par\bigskip\noindent
\textit{This work is dedicated to the memory of Vladimir Zakharov, a giant of nonlinear science.}
\end{abstract}


\bigskip
\section{Introduction and outline} 
\label{s:intro}

This work addresses the formulation of the inverse spectral theory for the non-self-adjoint Zakharov-Shabat (ZS) operator with periodic potentials using a Riemann-Hilbert problem (RHP) approach,
as well as the application of these results to solving the initial value problem (IVP) for the focusing nonlinear Schr\"odinger (NLS) equation under periodic boundary conditions.
Specifically, we study the eigenvalue problem associated with the ZS operator:
\vspace*{-1ex}
\be
\L\,v = z\,v\,,
\label{e:Diraceigenvalueproblem}
\ee
where $\L$ is a matrix-valued Dirac operator acting in $L^2(\Real,\C^2)$ with dense domain $H^{1}(\Real,\C^2)$, defined as
\vspace*{-1ex}
\be
\L:= \i\sigma_{3}(\partial_{x} - Q)\,
\label{e:Diracoperator}
\ee
with $v = v(x,z) = (v_1, v_2)^{T} $ and the superscript $T$ denoting matrix transpose.
Hereafter,
$\sigma_3 = \diag(1,-1)$ is the third Pauli matrix 
and $Q(x)$ is the matrix-valued potential function
\be
Q(x) = \begin{pmatrix} 0 & q(x) \\ -q^*(x) & 0 \end{pmatrix}\,,
\label{e:Q}
\ee
with the asterisk denoting complex conjugation.
The ``potential'' $q: \R \to \C$ is a function with minimal period~$L$,
i.e.,
\be
q(x+L) = q(x)\,,\quad \forall x\in\Real\,.
\label{e:Qperiodic}
\ee

The Lax spectrum of $\L$, denoted by $\Sigma(\L)$, is defined as the set of complex values $z$ for which there exists at least one bounded, nontrivial solution to the eigenvalue problem \eqref{e:Diraceigenvalueproblem}, i.e.,
\be
\Sigma(L) := \big\{z\in\C: \mathcal L\phi=z\phi,\,\, 0<\|\phi\|_{L^{\infty}(\R;\C^2)} < \infty \big\}\,.
\label{e:laxspec}
\ee
In this work, we will consistently use the term ``spectrum'' as a synonym for the Lax spectrum.  
For simplicity, and to minimize technical complications, we assume $q\in C^2([0,L])$ throughout, unless explicitly stated otherwise. 

It is well known that the eigenvalue problem~\eqref{e:Diraceigenvalueproblem} is intimately related to the focusing NLS equation, namely the partial differential equation (PDE)
\be
\i q_t + q_{xx} + 2 |q|^2 q = 0\,,
\label{e:nls}
\ee 
where $q = q(x,t)$, 
and subscripts denote partial derivatives.
This is because the Lax pair for the focusing NLS equation consists of the following overdetermined linear system of ordinary differential equations (ODEs) \cite{ZS1972,ZS1973}:
\vspace*{-1ex}
\bse%
\label{e:NLSLP}
\begin{gather}
	v_x = X(x,t,z)\,v\,,\qquad X(x,t,z) = -\i z\sigma_3 + Q(x,t)\,,
	\label{e:zs}
	\\
	v_t = T(x,t,z)\,v\,,\qquad T(x,t,z) = -2\i z^2\sigma_3 + H(x,t,z)\,,
	\label{e:NLSLP2}
\end{gather}
\ese
where $H(x,t,z) = 2zQ - \i\sigma_3(Q^2-Q_x)$ and where $Q(x,t)$ has the same form as in~\eqref{e:Q}
but with $q(x)$ replaced by $q(x,t)$. 
That is, the compatibility condition $v_{xt} = v_{tx}$ of the system~\eqref{e:NLSLP} is equivalent to the focusing NLS equation~\eqref{e:nls}.
The first half of the Lax pair, namely~\eqref{e:zs}, is known as the Zakharov-Shabat scattering problem, 
the eigenvalue $z\in\C$ is also referred to as the spectral parameter and $v=v(x,t,z)$ as the scattering eigenfunction.
Importantly, the eigenvalue problem~\eqref{e:Diraceigenvalueproblem} is equivalent to \eqref{e:zs}. 

In 1967, Gardner, Greene, Kruskal, and Miura introduced direct and inverse spectral methods to solve the IVP for the Korteweg-de Vries equation \cite{GGKM}. Shortly after, Zakharov and Shabat \cite{ZS1972,ZS1973} extended this approach,
now commonly referred to as the inverse scattering transform (IST), to the NLS equation~\eqref{e:nls}, 
relying on the spectral theory of the Zakharov-Shabat problem \eqref{e:nls}. 
Shortly afterwards Ablowitz, Kaup, Newell, and Segur broadened these methods to a large class of integrable nonlinear PDEs~\cite{AKNS1974}. 
Key advances in IST followed, particularly the reformulation of the inverse problem as a matrix RHP \cite{AC1991,BealsCoifman,Deift1998,Its2003,NMPZ1984,Zhou1989}, 
which enabled the development of powerful asymptotic methods such as the Deift-Zhou technique \cite{DeiftZhou1991}, that has been widely used with success across various applications.

After the development of direct and inverse spectral methods and the IST for decaying potentials, the focus naturally extended to IVPs with periodic boundary conditions. The spectral theory for Hill’s equation (the time-independent Schrödinger equation with periodic potentials) and its applications to the KdV equation under periodic boundary conditions have been studied extensively since the 1970s, forming the foundation of the finite-genus formalism \cite{BBEIM,Dubrovin1975,ItsMatveev1,ItsMatveev2,ItsMatveev3,KappelerPoschel2010,Lax1975,MW1966,Marchenko1974,MarchenkoOstrovsky1975,Matveev2008,McKeanVanMoerbeke,Novikov1974,NMPZ1984}. This framework was later expanded to the infinite-genus case, where global transformations map the KdV flow onto a periodic trajectory on an infinite-dimensional torus \cite{KappelerPoschel2010,McKeanTrubowitz1,McKeanTrubowitz2}.
Similarly, the spectral theory of focusing and defocusing Zakharov-Shabat operators and their finite-gap formulation has a longstanding research history \cite{GesztesyHolden,gesztesyweikard_acta1998,gesztesyweikard_bams1998,ItsKotlyarov1976,Kotlyarov1976,MaAblowitz,McKean1981,McLaughlinOverman}. 
Recently, McLaughlin and Nabelek developed a Riemann-Hilbert approach to solve the inverse spectral problem for Hill's operator in \cite{McLaughlinNabelek}. 
In \cite{BiondiniZhang2023} we then extended their approach to the self-adjoint Dirac operator 
[namely, the operator $\L$ that one would obtain when the minus sign in front of the term $q^*(x)$ in~\eqref{e:Q} is absent]
with periodic potentials, 
and we applied the results to solve the IVP for the defocusing NLS equation 
[namely, the PDE obtained when the nonlinear term in \eqref{e:nls} is preceded by a minus sign]
with periodic boundary conditions. 

In the present work, we further generalize the above approach to the non-self-adjoint Dirac operator~\eqref{e:Diracoperator} with periodic potentials, which enables us to solve the IVP for the focusing NLS equation~\eqref{e:nls} with periodic boundary conditions. 
The most significant difference between the spectral theory for the self-adjoint and non-self-adjoint Dirac operator is that, 
in the former, the spectrum is confined to the real $z$-axis.
In the latter, in contrast, the spectrum typically has non-real components, 
which introduces several complications in the analysis. 
In this work we show in detail how one can nonetheless address these challenges
and successfully generalize the formalism.

The outline of this work is as follows. 
In section \ref{s:spectrum} we formulate the direct spectral theory.  
We begin by recalling some results from Bloch-Floqeuet theory.
We then introduce a transformation of the fundamental matrix solution $Y(x,z)$ of the ZS problem, formulate some key assumptions,
and define the spectral data that will allow one to uniquely recover the potential $q(x)$,
including the Dirichlet spectrum.
In section \ref{s:asymptotics} we define the Bloch-Floquet eigenfunctions $\psi^\pm(x,z)$ of the modified ZS problem 
in terms of its modified fundamental matrix solution $\~Y(x,z)$
and we compute their asymptotic behavior as $z\to\infty$. 
Moreover, we express various quantities 
as infinite products via Hadamard's factorization theorem,
which will be instrumental for the inverse problem.
In section \ref{s:inverse} we present the formulation of the inverse problem as a RHP, for which we prove uniqueness of solutions, 
and we obtain the reconstruction formula for the potential based on the solution of the RHP. 
A crucial step in this process is introducing a matrix-valued function $B(z)$, which cancels the singularities of the RHP and ensures the uniqueness of the solution.
In section \ref{s:time} we compute the time dependence of the Dirichlet eigenvalues and Bloch-Floquet eigenfunctions and we 
show how the results of section \ref{s:inverse} can used to solve the IVP for the focusing NLS equation with periodic BC. 
Finally, in section~\ref{s:conclusions} we end this work with some concluding remarks.
The appendix is devoted to a few preliminary results and technical lemmas,
the explicit formulation of the RHP in a few special cases (namely, genus-zero and genus-one potentials) as examples,
a discussion of an efficient numerical method for the calculation of the Dirichlet spectrum, and some illustrative numerical results.

\section{Direct spectral theory for the periodic non-self-adjoint Zakharov-Shabat problem}
\label{s:spectrum}

In this section we begin formulating the direct spectral theory for the defocusing ZS spectral problem and define the spectral data that will allow us to uniquely reconstruct the potential in section~\ref{s:inverse}.
Since the time dependence of the potential does not play any role in the direct and inverse spectral theory,
in this section and the following ones we will temporarily omit the time dependence, 
which will then be restored when discussing the IVP for the NLS equation in section~\ref{s:time}.

\subsection{Lax spectrum, main eigenvalues and Dirichlet eigenvalues}

\paragraph{Fundamental analytic solutions, monodromy matrix, Floquet discriminant.}
We begin by briefly recalling some well-known results from Bloch-Floquet theory,
Recall that, for an $n \times n$ matrix-valued function $A\in L^1_{{\rm loc}}(\R)$ with $A(x+L) = A(x)$,
Floquet's theorem~\cite{Floquet_int,Eastham,Floquet} states that
any fundamental matrix solution $Y(x)$  of the system of linear homogeneous ODEs 
\vspace*{-0.5ex}
\be
y_x = A(x)\,y
\label{e:yprime} 
\ee
can be written in the Floquet normal form
\vspace*{-0.5ex}
\be
Y(x) = W(x)\,\e^{Rx}\,,
\label{e:canonical}
\ee
where $W(x)$ is a nonsingular matrix with $W(x+L) = W(x)$, and $R$ is a constant matrix. 
Thus, all bounded solutions of the ZS system~\eqref{e:zs} have the form
%
$v(x,z) = \e^{\i\nu x} w(x,z)$
(or linear combinations of such functions),
where 
$w(x+L,z) = w(x,z)$,
and $\nu\in\Real$ is the Floquet exponent (also referred to as quasi-momentum in some works). 
One also defines the so-called Bloch-Floquet solutions, or normal solutions, 
as the solutions of \eqref{e:zs} such that 
\vspace*{-0.5ex}
\be
\psi(x+L,z) = \rho\,\psi(x,z)\,,
\label{e:bdsol}
\ee 
where $\rho$ is the Floquet multiplier. 
Thus, 
a solution of~\eqref{e:zs} is bounded if and only if $|\rho|=1$, 
in which case $\rho= \e^{\i\nu L}$ with $\nu\in\Real$.
Moreover, 
the Floquet multipliers are the eigenvalues of the monodromy matrix $M(z)$,  
which is defined by
\vspace*{-0.5ex}
\be
Y(x+L,z)=Y(x,z)M(z)\,, 
\label{e:monodromy}
\ee
where $Y(x,z)$ is any fundamental matrix solution of~\eqref{e:zs}.
Hereafter, we choose $Y(x,z)$ as the principal matrix solution of~\eqref{e:zs}
that is,
the matrix solution of~\eqref{e:zs} normalized so that
$Y(0,z)\equiv\mathbb{I}$, where $\mathbb{I}$ is the identity matrix.
We then have
\vspace*{-0.5ex}
\be
M(z) = Y(L,z)\,.
\label{e:monodromy2}
\ee
Standard techniques allow one to show that, under the above assumptions, 
$Y(x,z)$ can be expressed as the following Volterra integral equation
\vspace*{-0.5ex}
\be
Y(x,z)=\e^{-\i z\sigma_3x}+\int_0^x\e^{-\i z\sigma_3(x-\xi)}Q(\xi)Y(\xi,z)\,\d\xi\,,
\label{phi}
\ee
which also allows one to show that, for all $x\in\Real$, $Y(x,z)$ is an entire function of $z$.
Since the RHS of ~\eqref{e:zs} is traceless, 
Abel's formula implies $\det M(z) \equiv 1$.  
Hence the eigenvalues of $M(z)$, i.e., the Floquet multipliers, are given by roots of the quadratic equation
\vspace*{-0.5ex}
\be
\rho^2 - 2\D(z)\rho+ 1 = 0,
\label{tracepoly}
\ee
where $\D(z)$ is the Floquet discriminant
\vspace*{-0.5ex}
\be
\D:=\D(z)=\half\tr M(z)\,,
\label{e:DeltaM}
\ee
which is also is an entire function of $z$ \cite{ForestLee,MaAblowitz, McLaughlinOverman}.
The eigenvalue problem \eqref{e:Diraceigenvalueproblem} has bounded solutions if and only if the following two conditions are simultaneously satisfied:
\vspace*{-0.5ex}
\be
\label{conditons4laxspec}
\Im \Delta(z)=0,\quad -1\leq \Re\Delta(z)\leq 1.
\ee
Thus, an equivalent representation of the Lax spectrum is as follows:
\vspace*{-0.5ex}
\be
\Sigma(\mathcal L) = \{ z \in \Complex : \D(z) \in [-1,1]\}.
\label{e:laxspec2}
\ee

\paragraph{Spectral bands and gaps.}
Equation~\eqref{e:laxspec2} implies that the Lax spectrum is a subset of the contour 
$\Gamma:=\{z\in\Complex: \Im \Delta(z)=0\}$ that is 
composed of the real $z$-axis plus a countable set of curves in the complex $z$-plane.
Each of these curves is a single Schwarz-symmetric analytic curve starting from infinity and ending at infinity in the complex plane.
The second condition in \eqref{conditons4laxspec}, namely, $-1\leq\Re\Delta\leq 1$, then divides each of the curves into a number (zero, one or more) of spectral bands. 
Denoting by $\Gamma_n$, for $n\in\mathbb{N}$, 
those analytic curves that contain a non-zero number of spectral bands,
we have 
\be
\label{e:Gamma}
\Gamma= \Real\cup\bigg(\bigcup_{n\in\Integer}\Gamma_n\bigg)\cup\Gamma_\mathrm{res},
\ee
where $\Gamma_\mathrm{res}$ is the union of all the analytic curves that do not contain any spectral bands.
Each finite portion of\, $\Gamma_n$ where $|\Re\Delta|\geq1$ and that is delimited by a band on either side is called a spectral gap.
Throughout this work, we will take the natural orientation of the real $z$-axis (i.e., oriented according to increasing values of $\Re z$) 
and we will take each $\Gamma_n$ to be oriented so that $\Im z$ is increasing in a neighborhood of infinity. 
With the above definition, one can rigorously talk about bands and gaps, as in a self-adjoint problem. The difference is that the bands and gaps are not restricted to lie on the real $z$-axis, but lie instead along curves $\Gamma_n$. Moreover, different spectral bands can intersect at saddle points of $\Delta$. 
Note also that no arc $\Gamma_n$ can be closed, and that  
each arc $\Gamma_n$ can intersect with another arc at most once.
(This because, in either case, one would obtain a closed contour on which $\Im\D(z)=0$, which would imply 
$\Im\D(z)\equiv0$ $\forall z\in\C$, which is a contradiction.)

\noindent
\begin{proposition}
If $q\in C^2(\Real)$, the Lax spectrum $\Sigma(\L)$ has the following properties
(e.g., see \cite{gesztesyweikard_acta1998,BOT_JST2023}):
\vspace*{-1ex}
\begin{itemize}
\advance\itemsep-4pt
\item It consists of a countable number of regular analytic arcs, the so-called spectral bands. 
\item The real $z$-axis is an infinitely long spectral band, i.e., $\Real\subset\Sigma(\L)$.
\item With the exception of the real axis, no spectral bands extend to infinity. 
\end{itemize}
\end{proposition}

For $z\in\Complex$, we denote the roots of~\eqref{tracepoly} as
\vspace*{-1ex}
\be
\rho_{1,2}(z) = \D(z) \mp (\D^2(z)-1)^{1/2},
\label{rho}
\ee
with $\rho_1$ and $\rho_2$ associated respectively with the plus and minus sign in~\eqref{rho},
for some appropriate choice of the complex square root, 
to be defined next.
Obviously $\rho_{1,2}(z)$ satisfy the relation $\rho_1(z) = 1/\rho_2(z)$.
For each $z\in\Sigma(\L)$, the condition $\D(z)\in[-1,1]$ implies $|\rho_{1,2}(z)|=1$, 
whereas for each $z\in\Complex\setminus\Sigma(\L)$, the condition $\D(z)\notin[-1,1]$ implies $|\rho_{1,2}(z)|\ne1$. 
It is also convenient to introduce the quantities
\be
r_\alpha(z) = (\Delta^2(z)-1)^\alpha,\qquad 
\alpha\in\R\,.
\ee
In particular, we can uniquely define the complex square root $r_{1/2}(z)$ in \eqref{rho} such that:
(i) its branch cut coincides with $\Sigma(\L)$;
(ii)
$r_{1/2}(z)\sim \Delta(z)$ as $|\Im z|\to\infty$; 
(iii)
$\forall z\in\Real$, $r_{1/2}(z)$ is continuous from above.
Moreover, $\forall z\in\Sigma(\L)\setminus\Real$, $r_{1/2}(z)$ is continuous from the left. 
%
We also define $\rho(z)$ be the root of~\eqref{tracepoly} that is analytic for $z\notin\Sigma(\L)$ and such that 
$|\rho(z)|<1$ for $z\in\C\setminus\Sigma(\L)$.
With the above definitions, we have $\rho(z) = \rho_1(z)$ $\forall z\notin\Sigma(\L)$.

\paragraph{Dirichlet spectrum.}
As is well known, knowledge of the Lax spectrum is not sufficient to uniquely reconstruct the potential.
Following \cite{MaAblowitz,McLaughlinOverman},
we also define the \textit{Dirichlet spectrum} associated with \eqref{e:zs} as follows:
\be
\Sigma_{\Dir}(x_o)
:=\{\mu \in \C : \exists~v \not\equiv 0 \in H^1([x_o,x_o+L],\C^2)~ ~\text{s.t.}~ ~ \L\,v = \mu\,v ~ \wedge~ v\in {\rm BC}_{{\rm Dir},x_o}
\}\,,
\label{e:dirichlet}
\ee
where ``BC$_{{\rm Dir},x_o}$'' denotes the following Dirichlet boundary conditions (BC) with base point $x_o$:
\be
\label{e:Dirbcs}
v_1(x_o)+v_2(x_o)=v_1(x_o+L)+v_2(x_o+L)=0\,,
\ee
with $v = (v_1,v_2)^T$.
Any point $\mu\in\Sigma_{\Dir}(x_o)$ will be referred to as a Dirichlet eigenvalue of~\eqref{e:Diraceigenvalueproblem}.
Similarly to the Floquet spectrum,
one can identify $\Sigma_{\Dir}(x_o)$ with the zero set of a suitable entire function.
For the ZS problem~\eqref{e:zs}, however,
additional complications arise compared to Hill's operator
(as was already shown in the self-adjoint case in \cite{BiondiniZhang2023}).
For this reason, 
we introduce the following similarity transformation, which will be instrumental not only for characterizing the Dirichlet spectrum,
but also for carrying out the inverse spectral theory:
\be
\~Y(x,z)=U\,Y(x,z)\,U^{-1},
\label{e:Ytildedef}
\ee
where
\be
U = \frac1{\sqrt2} \begin{pmatrix} 1 & 1 \\ -\i & \i \end{pmatrix}.
\label{e:Udef}
\ee
Then $\~Y(x,z)$ is the fundamental solution of the following modified scattering problem
\be
\~y_x = U(-\i z\sigma_3 + Q)U^{-1}\,\~y\,.
\label{e:modZS}
\ee
We also define
$\tilde{M}(z)=\~Y(L,z)$,
which is also an entire function of~$z$.
Moreover, since the trace and determinant are invariant under the transformation~\eqref{e:Ytildedef},
the Floquet eigenvalues of the modified scattering problem~\eqref{e:modZS},
i.e., the eigenvalues of $\~M(z)$, coincide with $\rho(z)$.
On the other hand, we now have:

\begin{proposition}
\label{p:dirzeros}
The Dirichlet spectrum $\{\mu_j\}$ with base point $x_o=0$ coincides with the set
\be
\Sigma_{\Dir}(0)= \left\{\mu \in \C : \~y_{12}(L,\mu) = 0\right\}.
\label{e:dirzeros}
\ee
\end{proposition}

\subsection{Further properties of the spectrum. Spectral data}

The formulation of the inverse spectral theory in section~\ref{s:inverse}
will require the use of some known results from spectral theory, 
which are summarized for convenience in the following theorem
\cite{grebertkappelermityagin,kappeler,McLaughlinOverman,MaAblowitz}
(a proof of some of these results is also given in Appendix~\ref{a:spectrum}):  
\begin{theorem}
\label{spectraprop}
The main spectrum and Dirichlet spectrum of the Dirac operator $\L$~\eqref{e:Diracoperator} with $L$-periodic potential 
$q\in L^2(\R,\C)$ satisfy the following properties:
\begin{enumerate}
\item 
The main spectrum $\{\z_j\}_{j\in\Integer}$ is defined as the roots of the equation $\Delta^2(z)-1=0$. 
The main eigenvalues\, $\zeta_j$ are either real or arise in complex conjugate pairs. 
\item 
Each real eigenvalue $\zeta_j$ is also a Dirichlet eigenvalue.
\item 
There exists an integer $J>0$ such that the entire function $\~y_{12}(L,z)$ has exactly one root in each disc:
\be\label{e:Dj}
D_j=\left\{z\in\Complex: |z-j\pi|<\frac{\pi}{4}\right\},\qquad |j|>J.
\ee
and exactly $2J+1$ roots in the disc $B_J=\{z\in\Complex:|z|<J\pi+\pi/4\}$ when counted with their multiplicities. 
There are no other roots.
\item 
(Counting Lemma) 
There exists an integer $J>0$ such that the entire function $\Delta^2(z)-1$ has exactly two roots in each disc $D_j$ with $|j|>J$ 
and exactly $4J+2$ roots in the disc $B_J=\{z\in\Complex:|z|<J\pi+\pi/4\}$ when counted with their multiplicities. There are no other roots.
\item 
If $q\in C^2(\Real)$, 
the periodic, antiperiodic and Dirichlet eigenvalues $\z_j$ and $\mu_j$  have the following asymptotic behavior as $|j|\to\infty$:
\be
\z_{2j},\z_{2j-1},\mu_j = \frac{\pi j}{L}+O\left(\frac{1}{j}\right).
\label{asymmuz}
\ee
\end{enumerate}
\end{theorem}

The proofs of the above ``counting lemmas'' can be found in \cite{gesztesyweikard_acta1998,kappeler}, 
and are straightforward applications of Rouche's theorem. 
Importantly, these results induce a natural two-to-one map between the main eigenvalues $\zeta_j$ and the Dirichlet eigenvalues $\mu_j$.

\begin{assumption}
\label{assump}
Throughout the rest of this work, we will make the following assumptions for the main spectrum and Dirichlet spectrum:
(i)
   Every analytic arc $\Gamma_n$ that contains one or more spectral bands intersects the real axis. 
(ii)
    All complex main eigenvalues are simple.
(iii)
    All the Dirichlet eigenvalues $\mu_j$ are simple zeros of $\~y_{12}(L,z)=0$,
    except those which coincide with real zeros of $\Delta^2-1$ with multiplicity $m$ higher than 2, which have multiplicity $m/2$.
\end{assumption}


(Regarding the last item of Assumption~\ref{assump},
we note that all real main eigenvalues are zeros of $\Delta^2-1$ with even multiplicity, since the real axis is one infinitely long band.)
Without loss of generality, we can label the points of the main spectrum such that $\zeta_{2j-1}=\zeta_{2j}^*$ with $\Im \zeta_{2j}\geq 0$. 
Moreover, employing the lexicographic ordering of complex numbers defined by
\be
a\preccurlyeq b :=\begin{cases}
	\Re a< \Re b,\\
	\Re a=\Re b ~\wedge~ \Im a\leq \Im b, 
	\end{cases}
\ee
we can take $\zeta_{2j}\preccurlyeq \zeta_{2j+2}$ for $\forall j\in\Integer$. 
Given Theorem \ref{spectraprop} and Assumption \ref{assump}, when $\z_{2j} = \z_{2j-1}$
(implying that $\z_{2j}$ and $\z_{2j-1}$ are real) we sort the Dirichlet eigenvalues so that $\mu_j = \z_{2j}$.
Note that we cannot say whether $\mu_j\preccurlyeq \mu_{j+1}$ for $\forall j\in\Integer$. 
However, by Theorem \ref{spectraprop},  there exists a $J\in\mathbb{N}$, s.t., $\mu_j\preccurlyeq \mu_{j+1}$ for $\forall |j|>J$.

\begin{lemma}
\label{y2}
At each value $z = \mu_j$,  one has 
$\~y_{22}(L,\mu_j) = 1/\~y_{11}(L,\mu_j)$, with one of them equal to $\rho(\mu_j)$ or $\rho^{-1}(\mu_j)$, 
and the vector-valued fundamental solution $\~y_2(x,\mu_j)$ is a Bloch-Floquet solution of~\eqref{e:modZS}
with Floquet multiplier 
 $\~y_{22}(L,\mu_j)$.
\end{lemma}

\begin{proof}
The Dirichlet eigenfunction  $\~y_2(x,\mu_j)$ solves the modified ZS problem with spectral parameter $z=\mu_j$, 
and so does $\~y_2(x+L,\mu_j)$.  
So we have $\~y_2(x+L,\mu_j)=a\~y_1(x,\mu_j)+b\~y_2(x,\mu_j)$ for some constants $a$ and $b$. 
Evaluating this expression at $x=0$, we conclude $a=0$ and $b=\~y_{22}(L,\mu_j)$. 
Therefore, by Floquet's theorem, $\~y_2(x,\mu_j)$ is a Bloch-Floquet solution with multiplier $\~y_{22}(L,\mu_j)=\rho(\mu_j)$  or $\~y_{22}(L,\mu_j)=\rho(\mu_j)^{-1}$. Since $\det \~Y(L,\mu_j)=1$, we also have $\~y_{22}(L,\mu_j)= 1/\~y_{11}(L,\mu_j)$.
\end{proof}

As in the defocusing case, the main spectrum $\{\zeta_j\}_{j\in\Integer}$ 
decomposes into nondegenerate band edges $E_j$ and degenerate band edges $\hat{\zeta}_j$.
By Theorem~\ref{spectraprop} and Assumption~\ref{assump}, each $E_j$ is a simple root of $\Delta^2-1=0$, while each $\hat{\zeta}_j$s is a double root.
Let $\Sigma^o(\mathcal L) = \Sigma(\mathcal L)\setminus\{\zeta_j\}_{j\in\Integer}$ denote the interior of the Lax spectrum
minus all the double main eigenvalues. 
An important distinction between the defocusing and focusing cases is that, in the former, the Dirichlet eigenvalues are confined to 
the (degenerate or non-degenerate) band gaps, whereas in the latter they can also lie in the interior of the spectral bands.
(We will show examples of this.)

\begin{definition}
\label{nudef}
For each Dirichlet eigenvalue $\mu_k\notin\Sigma^o(\mathcal L)$, we define $\nu_k=-\sgn(\log|\~y_{22}(L,\mu_k)|)$,
with the signum function taken to be zero when its argument is zero.
For each $\mu_k\in\Sigma^o(\mathcal L)$ 
we take $\nu_k=1$ if $\~y_{22}(L,\mu_k) = \rho(\mu_k)$ and $\nu_k = -1$ if $\~y_{22}(L,\mu_k) = 1/\rho(\mu_k)$.
With this definition, $\nu_k$ is such that $\~y_{22}(L,\mu_k) = \rho^{\nu_k}(\mu_k)$.
\end{definition}
\begin{proposition}
For any Dirichlet eigenvalue $\mu_j$, $\nu_j=0$ if and only if $\mu_j$ sits at a band edge, i.e., $\mu_j \in \{\zeta_k\}_{k\in\Integer}$.
\end{proposition}
\begin{proof}
Note first that, since $\det\~Y(x,z)=1$, whenever 
$\~y_{12}(L,\mu_j) =0$ one has $\~y_{11}(L,\mu_j)\~y_{22}(L,\mu_j) = 1$.
The first part of the result then follows 
because, 
at a band edge, $\~y_{11}(L,z) + \~y_{22}(L,z) = \pm2$,
and the above two conditions can only be simultaneously satisfied if $\~y_{11}(L,\mu_j) = \~y_{22}(L,\mu_j)$ 
and both are equal respectively to 1 or $-1$, implying $\nu_j=0$.
The second part of the result follows because, by Definition~\ref{nudef}, $\nu_j = 0$ can only happen when $\mu_j\notin\Sigma^o(\L)$.
But since $|\rho(z)|<1$ for all $z\notin\Sigma(\L)$, one can only have $|\~y_{12}(L,\mu_j)|=1$ if $\mu_j\in\Sigma(\L)$.
Therefore, $\mu_j\in\Sigma(\L)\setminus\Sigma^o(\L) = \{\zeta_k\}_{k\in\Integer}$.
\end{proof}

\begin{definition}(Fixed and movable Dirichlet eigenvalues)\label{def:dirichleteigenvalue}
Let $\mu\in\C$ be a Dirichlet eigenvalue associated to the monodromy matrix $\tilde{M}(z; x_o)$ with a given base point $x = x_o$, Then we say that $\mu$ is a fixed Dirichlet eigenvalue if
$\mu \in \Sigma_{\Dir}(L;x)$ for all $x\in\R$. Otherwise, we say $\mu\in\C$ is a movable Dirichlet eigenvalue. 
We use $\mucirc_k$ to denote either a movable Dirichlet eigenvalue or an element of $\Sigma^o(\mathcal L)$\,.
\end{definition}

We are now ready to define the spectral data that will uniquely determine the potential~$q(x)$.
Specifically, excluding 
the trivial potential $q(x)\equiv0$ from consideration, 
we define the spectral data associated to the potential $q$ as the set
\be
S(q):=\{{E}_{2k-1},{E}_{2k},\mucirc_k,\nu_k\}_{k=\g_-,\dots,\g_+}
\label{e:spectraldata}
\ee
Next we introduce a few additional quantities that will appear in the formulation of the Riemann-Hilbert problem in section~\ref{s:inverse}.

\begin{definition}
\label{def:varpi_gamma}
Let $m_n$ be the number of bands along $\Gamma_n$. Then we define $\varpi_0$ as some intersection in the following way:
	\vspace*{-1ex}
	\begin{enumerate}
	\item[(i)] If $m_0\!\!\mod 2\neq0$, let $n_0=0$ and $\varpi_0=\Gamma_0\cap\Real$.
	\item[(ii)] If $m_0\!\!\mod 2=0$, let $n_0$ be the smallest positive integer such that $m_0\mod 2\neq0$, and let $\varpi_0=\Gamma_{n_0}\cap\Real$.
	\end{enumerate}
Then let $n_1$ be the smallest integer greater than $n_0$ such that $m_{n_1}\mod 2\neq0$, and let $\varpi_1=\Gamma_{n_1}\cap\Real$.
 By induction, for $j\geq 2$, let $n_j$ be the smallest integer greater than $n_{j-1}$ such that $m_{n_j}\mod 2\neq0$, and $\varpi_j=\Gamma_{n_j}\cap\Real$. 
 Similarly, for $j<0$, let $n_{j-1}$ be the smallest integer less than $n_{j}$ such that $m_{n_j}\mod 2\neq0$, and $\varpi_j=\Gamma_{n_j}\cap\Real$.
Finally, let 
$\Gamma_n^\pm = \Gamma_n\cap \mathbb{C}^\pm$,
and let $g_n$ the number of non-degenerate main eigenvalues on $\Gamma_n^+$
Additionally, for all $n\in\Integer$, let 
\vspace*{-0.6ex}
\[
\Gamma_n^+=\bigcup\limits_{k=0}^{g_n}\Gamma_{n,k}^+,
\]
where $\Gamma_{n,k}^+$ is taken to denote either a band or a gap on $\Gamma_n^+$ depending on whether $k$ is respectively odd or even, 
with the first band from the top to bottom denoted as $\Gamma_{n,1}^+$,
and $\Gamma_{n,0}^+$ is the portion of $\Gamma_n^+$ extending out to infinity.
Correspondingly, in the lower half-plane we set $\Gamma_{n,k}^-=\Gamma_{n,k}^{+*}$.
\end{definition}

\begin{figure}[t!]
\centerline{\includegraphics[width=0.55\textwidth]{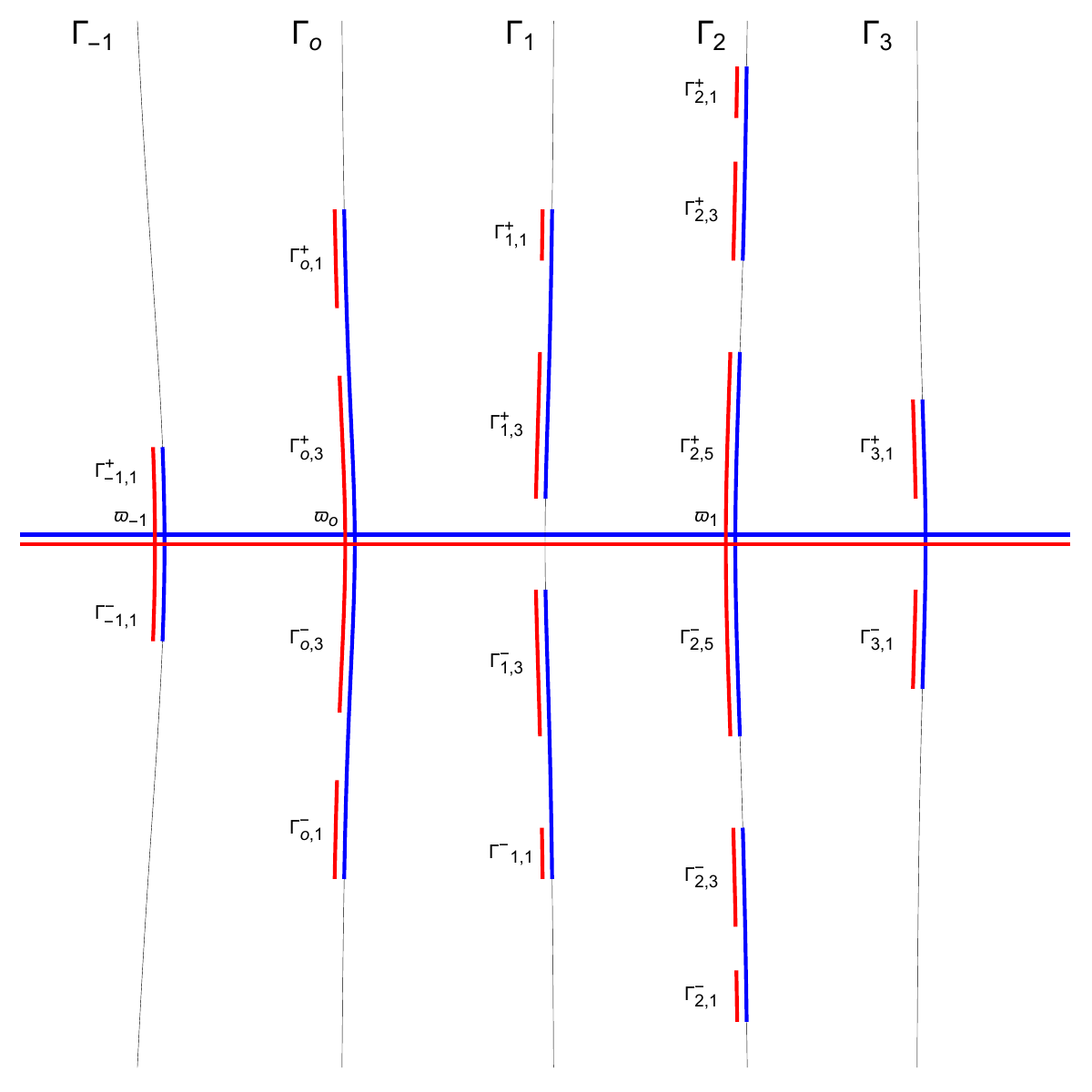}}
\caption{Schematic diagram showing a particular realization of the analytic arcs $\Gamma^\pm_{n,k}$,
non-degenerate main spectrum and the choice of branch cuts for $r_{1/2}(z)$ (in red) and for $r_{1/4}(z)$ (in blue).}
\label{f:fourthrootbranchcuts}
\end{figure}

Figure~\ref{f:fourthrootbranchcuts} is a cartoon showing a schematic representation of the spectrum, the analytic arcs $\Gamma_n$, etc.
Concrete examples of the class of potentials satisfying the above assumptions
(and which can therefore be treated with the present formalism)
include the genus-zero constant potential $q(x) = A\,\e^{i\alpha}$
with $A>0$ and $\alpha\in\Real$,
as well as $q(x) = A\,\dn(x,m)$ 
(where $\dn(x,m)$ is one of the Jacobi elliptic functions 
and $m$ is the elliptic parameter \cite{NIST}),
which was studied in \cite{AdvMath2023}, 
where it was shown that it is a finite-genus potential (with genus equal to $2A-1$) if and only if $A\in\Natural$.
We also believe that most of the finite-genus potentials of the focusing NLS equation \cite{BBEIM} will also satisfy the assumptions. 
See also \cite{DeconinckSegal} for other genus-one potentials of the focusing NLS equation.
The example a constant (genus-zero) potential will be treated in detail in Appendix~\ref{s:genuszero},
whereas the example of the genus-one potential $q(x) = \dn(x,m)$ will be treated in detail in Appendix~\ref{s:dn}. 

We end this section with a result that will allow us to simplify the formulation of the Riemann-Hilbert problem for a large class of potentials:

\begin{lemma}
If the potential $q(x)$ is real and even,
the Dirichlet eigenvalues with base point $x_o=0$ are located at points of the main spectrum.
\label{l:Dirichletbase0}
\end{lemma}

\begin{proof}
Recall first that, for any real and even potential, the symmetries of the ZS problem allow the monodromy matrix $M(z)$ to be written as 
\be
M(z) = \Delta(z) I + c(z) \sigma_3 - \i s(z)\sigma_2\,,
\label{e:MDeltaIcs}
\ee
where $I$ and $\sigma_2$ are the identity matrix and the second Pauli matrix (e.g., see \cite{AdvMath2023}).
In terms of $M(z)$, the Dirichlet eigenvalues are the zeros of 
$\frac \i2 [ M_{11} - M_{22} +   M_{21} - M_{12} ] = \i(c(z) + s(z))$.
Thus, the Dirichlet eigenvalues are the points for which $c(z)=-s(z)$.
Next, note that, since $\det M(z) =1$, \eqref{e:MDeltaIcs} yields
\be
\Delta^2(z)= 1 + c^2(z) - s^2(z)\,.
\ee
The main spectrum is therefore located at the points for which $c^2(z)=s^2(z)$.  
We then see that each Dirichlet eigenvalue must coincide with a main eigenvalue.
\end{proof}

\section{Modified Bloch-Floquet solutions, asymptotics and infinite product expansions}
\label{s:asymptotics}

Some additional results are needed in order to formulate the inverse spectral theory.
Specifically, we need to control the asymptotic behavior of relevant quantities as $z\to\infty$.
\begin{lemma}
\label{l:asymY}
If $q\in C^2([0,L],\Complex)$, then for all $z\in\Complex$ the fundamental matrix $Y(x,z)$ has
the following asymptotic behavior as $z\to\infty$:
\be
\label{asymY}
Y(x,z)= \left[ I + \frac1{2iz} \sigma_3 K[q](x) 
  + \frac1{2iz} \begin{pmatrix} 0 & \!\! q(x) - q(0)\,\e^{-2izx} \\  q^*(x) - q^*(0)\,\e^{2izx} \!\! &0 \end{pmatrix} 
  + O(1/z^2) \right] \,\e^{-\i zx\sigma_3}\,,
\ee
where for brevity we defined 
\be
K[q](x) =  \int_0^x|q(\xi)|^2\d\xi\,.
\ee
\end{lemma}
The proof of Lemma~\ref{l:asymY} is obtained following identical methods as in \cite{BOT_JST2023,BiondiniZhang2023}, and is therefore omitted for brevity.
As already mentioned earlier, for simplicity in this work we will always assume that $q\in C^2([0,L],\Complex)$ unless explicitly stated otherwise.
It is straightforward to show that the asymptotic behavior of $Y(x,z)$ in Lemma~\ref{l:asymY} implies the following:
\begin{proposition}
\label{asymD}
The Floquet discriminant $\D(z)$ has the following asymptotic behavior as $z\to\infty$:
\vspace*{-0.6ex}
\begin{gather}
\Delta(z)=\half\e^{\mp\i zL}\left(1\pm\frac{1}{2\i z}\int_0^L|q(x)|^2\d x+O\left(\frac{1}{z^2}\right)\right),
\qquad 
z\in\C^\pm\,.
\end{gather}
\end{proposition}

Next we introduce a family of Bloch-Floquet solutions of the modified ZS problem~\eqref{e:modZS},
that is instrumental in the formulation of the inverse problem:
\begin{definition}
Define $\psi^\pm(x,z)$ as the Bloch-Floquet solutions of the modified ZS problem~\eqref{e:modZS} 
uniquely determined by the conditions 
\be
\psi^{\pm}(x+L,z)=\rho^{\pm1}(z)\psi^{\pm}(x,z)\,,\qquad
\psi_1^{\pm}(0,z)=1.
\label{e:psipmdef}
\ee
Also, define $\Psi(x,z)$ as the matrix Bloch-Floquet, given by
\vspace*{-1ex}
\be
\label{e:defPsi}
\Psi(x,z) = \big( \,\psi^\mp\,, \,\psi^\pm\, \big)\,,\qquad z\in\Complex^\pm\setminus\Real\,.
\ee
\end{definition}
\noindent 
The following propositions are proved in an identical way as in the self-adjoint case \cite{BiondiniZhang2023}. 
We therefore omit the proof for brevity.

\begin{proposition}
\label{p:BFyexpansion}
The Bloch-Floquet solutions~ $\psi^\pm(x,z)$ defined by~\eqref{e:psipmdef} are given by
\be
\psi^{\pm}(x,z)=\~y_1(x,z)+\frac{\rho^{\pm1}(z)-\~y_{11}(L,z)}{\~y_{12}(L,z)}\~y_2(x,z),
\label{BF}
\ee
where $\~Y(x,z) = (\~y_1,\~y_2)$.
\end{proposition}

\begin{proposition}
\label{p:asympsi}
For all $x>0$, $\psi^-(x,z)$ has the following asymptotic behavior as $z\to\infty$:
\label{asympsi-}
\be
\psi^-(x,z) =  \begin{cases} \displaystyle 
    \left[ \begin{pmatrix} 1 \\ -\i \end{pmatrix}
     + \frac1{2iz} \begin{pmatrix} 1 \\ -\i \end{pmatrix} K[q](x)
     + \frac1{2iz} \begin{pmatrix} q^*(x) - q^*(0) \\ \i\,\big( q^*(x) +  q^*(0) \big) \end{pmatrix}
     + O(1/z^2) \right]\,\e^{-\i zx} \,,\quad z\in\C^+\,,
             \\ \displaystyle
     \left[ \begin{pmatrix} 1 \\ \i \end{pmatrix}
     - \frac1{2iz} \begin{pmatrix} 1 \\ \i \end{pmatrix} K[q](x)
     + \frac1{2iz} \begin{pmatrix} q(x) - q(0) \\ -\i\,\big( q(x) + q(0) \big) \end{pmatrix} 
     + O(1/z^2) \right]\,\e^{\i zx} \,,\qquad z\in\C^-\,,
     \end{cases}
\ee
with $K[q](x)$ as before.
Similarly, $\psi^+(x,z)$ has the following asymptotic behavior as $z\to\infty$:
\label{asympsi+}
\be
\psi^+(x,z) =  \begin{cases} \displaystyle 
    \left[ \begin{pmatrix} 1 \\ \i \end{pmatrix}
     - \frac1{2iz} \begin{pmatrix} 1 \\ \i \end{pmatrix} K[q](x)
     + \frac1{2iz} \begin{pmatrix} q(x) - q(0) \\ -\i\,\big( q(x) + q(0) \big) \end{pmatrix} 
     + O(1/z^2) \right]\,\e^{\i zx} \,,\quad z\in\C^+\,,
             \\ \displaystyle
    \left[ \begin{pmatrix} 1 \\ -\i \end{pmatrix}
     + \frac1{2iz} \begin{pmatrix} 1 \\ -\i \end{pmatrix} K[q](x)
     + \frac1{2iz} \begin{pmatrix} q^*(x) - q^*(0) \\ \i\,\big( q^*(x) +  q^*(0) \big) \end{pmatrix}
     + O(1/z^2) \right]\,\e^{-\i zx} \,,\quad z\in\C^-\,,
     \end{cases}
\ee
\end{proposition}

As in \cite{BiondiniZhang2023}, the significance of the Bloch-Floquet eigenfunctions is that their asymptotic behavior as $z\to\infty$ allows one to uniquely recover the potential $q(x)$.
Indeed, it is straightforward to see the asymptotic behavior in Proposition~\ref{p:asympsi}, 
together with the definitions~\eqref{e:Ytildedef} and~\eqref{BF}, yields:
\begin{corollary}
The matrix potential $Q(x)$ of the Dirac operator~\eqref{e:Diracoperator} can be recovered as 
\be
Q(x) = \lim_{z\to\infty} \i z\,[\sigma_3,\~U\, \Psi(x,z)\e^{\i zx\sigma_3}]\,.
\label{e:reconstruction}
\ee
where 
\be
\~U = (1/\sqrt2) \,U^{-1} = \frac12 \begin{pmatrix} 1 & \i \\ 1 & -\i \end{pmatrix}\,.
\label{e:Utildedef}
\ee
\end{corollary}

On the other hand, similarly to the self-adjoint case, \eqref{BF} implies that $\psi^\pm(x,z)$ are singular at the
Dirichlet eigenvalues $\mu_j$.  These singularities can either be removable or simple poles, as discussed in the following 
two propositions:

\begin{proposition}
If $\nu_j=0$, both $\psi^+(x,z)$ and $\psi^-(x,z)$ remain finite at $z=\mu_j$.
\label{p:psipm_nu=0}
\end{proposition}

\proof
When $\nu_j=0$, from Definition \ref{nudef}, we have $\rho(\mu_j)=1/\rho(\mu_j)=1$. By Lemma \ref{y2}, it follows that $\~y_{11}(L,\mu_j)=\rho(\mu_j)=1/\rho(\mu_j)=1$. This implies that $\mu_j$ is a zero of both $\~y_{12}(L,z)$ and $\rho^{\pm1}(z)-\~y_{11}(L,z)$.
Thus, by~\eqref{BF} and Condition (ii) of Assumption~\ref{assump}, $\psi^{\pm}$ remains finite at $\mu_j$. 
\endproof

\begin{proposition}
If $\nu_j=1$, then $\psi^-(x,\mu_j)$ is finite, while $\psi^+(x,z)$ has a simple pole at $\mu_j$.
Conversely, if $\nu_j=-1$, then $\psi^+(x,\mu_j)$ is finite, while $\psi^-(x,z)$ has a simple pole at $\mu_j$.
\end{proposition}
\begin{proof}
    We prove the case when $\nu_j=1$, the case with $\nu_j=-1$ follows analogously.  By Definition \ref{nudef}, we have $\~y_{22}(L,\mu_k)=\rho(\mu_k)=1/\~y_{11}(L,\mu_k)$, indicating that $\mu_j$ is a zero of both $\~y_{12}(L,z)$ and $\rho^{-1}(z)-\~y_{11}(L,z)$. Consequently, $\psi^-(x,z)$ is finite at $\mu_j$ and $\psi^+$ has a pole at $\mu_j$.
\end{proof}

The last ingredient needed in order to formulate the inverse problem are infinite product expansions
for various relevant quantities:

\begin{remark}
Using the asymptotics in~\eqref{asymY}
and the corresponding expressions for the entries of\, $\~Y(x,z)$
(which are omitted for brevity but can be obtained in a straightforward way from~\eqref{asymY}),
one can prove that $\Delta(z)$, and $\~y_{12}(L,z)$ are entire functions with order~1, as in the self-adjoint case~\cite{BiondiniZhang2023}.
Thus, Hadamard's factorization theorem (e.g., see \cite{Ahlfors,Krantz}), 
allows us to write $\Delta^2(z)-1$ and $\~y_{12}(L,z)$ as the following infinite products:
if $\mu_0\ne0$ and $\z_{0}\ne0$,
\bse
\label{e:Hadamardproducts}
\begin{gather}
\~y_{12}(L,z)=\e^{A_1z+B_1}\prod_{j\in\mathbb{Z}}\left(1-\frac{z}{\mu_j}\right)\e^{\frac{z}{\mu_j}}\,,
\label{y12}
\\
\Delta^2(z)-1=\e^{A_2z+B_2}\prod_{j\in\Z}\left(1-\frac{z}{\z_{2j}}\right)\left(1-\frac{z}{\z_{2j-1}}\right)\e^{z\big(\frac{1}{\z_{2j}}+\frac{1}{\z_{2j-1}}\big)}\,.
\label{e:d2-1}
\end{gather}
\ese
If $\mu_0=0$, \eqref{y12} is replaced by
\bse
\be
\~y_{12}(L,z)=z\e^{A_1z+B_1}\prod_{j\neq 0}\bigg(1-\frac{z}{\mu_j}\bigg)\,\e^{\frac{z}{\mu_j}}.
\ee
\ese
From \cite{kappeler} (specifically, see Lemma~5.4 and Appendix~C), the above canonical products could also be written as
\bse
\label{e:modifiedinfiniteproducts}
\begin{gather}
\~y_{12}(L,z)=-\prod_{j\in\Integer}\frac{\mu_j-z}{\pi_j},\label{y12simpler}\\
\Delta^2(z)-1=-\prod_{j\in\Integer}\frac{(\zeta_{2j}-z)(\zeta_{2j-1}-z)}{\pi_j^2},
\end{gather}
\ese
where $\pi_j$ is defined as
\be
\pi_j=\begin{cases}
	j\pi,& j\neq 0\\
	1, & j=0.
	\end{cases}
\ee
\end{remark}

\section{Inverse spectral theory for the periodic non-self-adjoint Zakharov-Shabat problem}
\label{s:inverse}

We are now ready to begin formulating the inverse spectral theory.  
The starting point is the matrix Bloch-Floquet eigenfunction $\Psi(x,z)$ in~\eqref{e:defPsi}.
One would like to formulate a Riemann-Hilbert problem for $\Psi(x,z)$.
As in the self-adjoint case, however, there are two complications:
(i)~The columns $\psi^\pm(x,z)$ of $\Psi(x,z)$ have a possibly infinite number of poles;
(ii)~$\det\Psi(x,z)\ne1$.
As in \cite{McLaughlinNabelek,BiondiniZhang2023}, both issues can be ``cured'' by introducing a suitable matrix $B(z)$ as discussed below.
An additional complication compared to the self-adjoint case, however, is that, in the non-self-adjoint case, the discontinuities of
$\Psi(x,z)$ are not limited to the real $z$-axis.

\begin{definition}
\label{d:f0f+f-}
Let the $2\times2$ matrix-valued function $B(z)$ be defined for all $z\in\Complex\setminus\Sigma(\L)$ as
\bse
    \be
	B(z) = b(z) \begin{cases} 
            \i\,\diag(f^-,f^+),&z\in\Complex^+\setminus\Sigma(\L)\,,\\[1ex]
            \diag(f^+,f^-),&z\in\Complex^-\setminus\Sigma(\L)\,,
	\end{cases}
	\label{B}
    \ee
    with
    \be
        b(z) = \frac{(f^0(z))^{1/2}}{(\Delta^2(z)-1)^{1/4}}\,,
    \ee
and where $f^0(z)$ and $f^\pm(z)$ are defined as 
    \be
	f^0(z)= - \prod_{\nu_j=0}\frac{\mu_j-z}{\pi_j}\,,\quad
	f^+(z)=\prod_{\nu_j=1}\frac{\mu_j-z}{\pi_j}\,,\quad
	f^-(z)=\prod_{\nu_j=-1}\frac{\mu_j-z}{\pi_j}\,.
	\label{e:f+-0def}
    \ee
\ese
\end{definition}

One can show that each of the infinite products in \eqref{e:f+-0def} is convergent for the same reasons that those in~\eqref{e:modifiedinfiniteproducts} are (e.g., see \cite{kappeler}).
Note also that 
\be
\~y_{12}(L,z) = f^0(z) f^+(z) f^-(z). 
\ee
Note also that, since $B(z)$ involves $(f^0(z))^{1/2}$ and $r_{1/4}(z) = (\D^2(z)-1)^{1/4}$, 
in order to  define it uniquely (and formulate the Riemann-Hilbert problem),
one must introduce a proper choice of branch cut for these quantities.  
We do so below.
In the meantime, however, we have:

\begin{proposition}
\label{p:spectralproducts}
The matrix $B(z)$ is completely determined by the spectral data $S(q)$ defined in \eqref{e:spectraldata}.
Explicitly:
\vspace*{-0.6ex}
\bse
\begin{align}
&b^4(z) =  - \prod\limits_{\nu_j=0} \frac{(\mucirc_j-z)^2}{(E_{2j}-z)(E_{2j-1}-z)} \prod\limits_{\nu_j\neq0}\frac{\pi_j^2}{(E_{2j}-z)(E_{2j-1}-z)}\,,
    \label{e:rootratio}
    \\
    &f^+(z)=\prod_{j=\g^-,\atop\nu_j=1 }^{\g^+}\frac{\mucirc_j-z}{\pi_j}\,,\qquad
    f^-(z)=\prod_{j=\g^-\atop\nu_j=-1 }^{\g^+}\frac{\mucirc_j-z}{\pi_j}\,.
    \label{e:fpmnew}
\end{align}
\ese 
\end{proposition}

\begin{proof}
By direct calculation, we have 
\vspace*{-1ex}
\begin{multline}
b^4(z) = \frac{
	\prod\limits_{\nu_j=0}(\mu_j-z)^2}{\prod\limits_{\nu_j=0}(\z_{2j}-z)(\z_{2j-1}-z)}\prod\limits_{\nu_j\neq0}\frac{\pi_j^2}{(E_{2j}-z)(E_{2j-1}-z)}
\\[-1ex]
    = \prod\limits_{\nu_j=0}\frac{(\mucirc_j-z)^2}{(E_{2j}-z)(E_{2j-1}-z)}
	\prod\limits_{\nu_j\neq0}\frac{\pi_j^2}{(E_{2j}-z)(E_{2j-1}-z)},
\end{multline}
where in the last equality we used the fact that all real Dirichlet eigenvalues coincide with the double main eigenvalues $\zeta_j$, allowing them to cancel each other in the first fraction. 
\end{proof}


\begin{proposition}
\label{l:Brealevenq}
The matrix $B(z)$ defined in \eqref{B} satisfies
	\be
	\det B(z)= \mp \frac{\~y_{12}(L,z)}{(\D^2-1)^{1/2}}\,, \qquad z\in\Complex^\pm\,.
	\ee
Moreover, if $q(x)$ is real and even, $B(z)$ is proportional to the identity matrix.
\end{proposition}

\begin{proof}
The first part of the result is obtained by straightforward calculation. 
For example, for $z\in\Complex^+$, we have 
$\det B(z)= - f^0(z)f^+(z)f^-(z)/(\Delta^2(z)-1) ^{1/2} = - \~y_{12}(L,z)/(\D^2(z)-1)^{1/2}$. 
The second part of the result is a consequence of Lemma~\ref{l:Dirichletbase0}.  
Indeed, the lemma immediately implies that,
if $q(x)$ is real and even, $f^\pm(z)\equiv 1$.
\end{proof}

Next we define the branch cuts for $(f^0(z))^{1/2}$ and $r_{1/4}(z)$ and 
introduce all the quantities that will eventually appear in the Riemann-Hilbert problem.

\begin{definition}
Let $m_n$ be the number of bands along $\Gamma_n$ and $g_n$ be the number of non-degenerate main eigenvalues as defined in Definition \ref{def:varpi_gamma}.
Without loss of generality, we fix the choice of branch for $r_{1/2}(z)$ such that $r_{1/4}(z)>0$ as $\Im z\to\infty$ on $\Gamma_0$.
Then we uniquely define the complex fourth root $r_{1/4}(z)$ so that:
\vspace*{-1ex}
\
\begin{enumerate}
\advance\itemsep-4pt
\item 
If $m_n=1$, then the whole band is a branch cut. If $m_n\neq1$, then $\Gamma_n\setminus(\cup_{k=0}^{[g_n/4]}\Gamma_{n,4k}^\pm)$ is the branch cut, where $[x]$ denotes the floor function. 
\item 
$r_{1/4}(z) =  \D^{1/2}(z)(1+o(1))$ as $|\Im z|\to\infty$.
\end{enumerate}
Further, let\,
$\Phi(x,z)$ be the sectionally analytic $2\times2$ matrix defined as
	\be
	\Phi(x,z) =  \frac12 U^{-1}\, \Psi(x,z)B(z)\e^{\i zx\sigma_3}\,,
	\quad z\in\Complex\,.
	\label{Phi}
	\ee 
Also, let $V(x,z)$ be the $2\times2$ jump matrix defined as 
	\begin{align}
	V(x,z) = \begin{cases}
    (-1)^{\mathcal{M}_n(z)+1}\begin{pmatrix}
			{f^-}/{f^+} & 0\\
			0 & {f^+}/{f^-}
		\end{pmatrix},
		& z\in\Omega\,,\\
    (-1)^{\mathcal{M}_n(z)}\begin{pmatrix}
			{f^-}/{f^+} & 0\\
			0 & {f^+}/{f^-}
		\end{pmatrix},
		 &z\in\Real\setminus\Omega\,,\\
    (-1)^{\mathcal{M}_n^+(z)}(-\i)^k\e^{-\i zx\hat\sigma_3}
		\begin{pmatrix}
			0 & {f^+}/{f^-}\\
			{f^-}/{f^+} & 0
		\end{pmatrix}, & z\in \Gamma_{n,k}^+,\quad $k$~\textrm{odd},\\
	(-1)^{\mathcal{M}_n^+(z)}(-\i)^k, & z\in \Gamma_{n,k}^+,\quad $k$~\textrm{even},\\
    (-1)^{\mathcal{M}_n^-(z)}\i^k\e^{-\i zx\hat\sigma_3}
		\begin{pmatrix}
			0 & {f^-}/{f^+}\\
			{f^+}/{f^-} & 0
		\end{pmatrix}, &  z\in\Gamma_{n,k}^-,\quad $k$~\textrm{odd},\\
	(-1)^{\mathcal{M}_n^-(z)}\i^k, & z\in \Gamma_{n,k}^-,\quad $k$~\textrm{even},
	\end{cases}
\label{e:Vexplicit}
\end{align}
where $\Omega = \cup_{k\in\mathbb{Z}}[\varpi_{2k},\varpi_{2k+1}]$ and 
$\e^{ia\hat\sigma_3}A = \e^{ia\sigma_3}A\,\e^{-ia\sigma_3}$ for any $2\times2$ matrix~$A$.



The quantities $\mathcal{M}_n^\pm$ in~\eqref{e:Vexplicit} are defined as follows. 
Let $s$ be the number of Dirichlet eigenvalues $\mucirc$ in $\Gamma_n^+$, with $\nu=0$ on the band edges. 
Label the intersection of $\Gamma_n^+$ and the real axis $\eta_{s+1}$ and label the uppermost band edge on $\Gamma_n^+$ as $\eta_0$. 
Starting at $\eta_{s+1}$ and moving towards infinity along $\Gamma_n^+$, label these $\mucirc$ as $\eta_s,\eta_{s-1},...,\eta_1$ in the order they appear along the curve. 
Define $\gamma_{np}$ as the partial arc of $\Gamma_n^+$ between $\eta_p$ and $\eta_{p+1}$. Then
\be
\label{mnplus}
    \mathcal{M}_n^+(z)=p, \qquad z\in \gamma_{np}.
\ee
Define $\mathcal{M}_n^-$in the same way as $\mathcal{M}_n^+$, but using $\Gamma_n^-$ instead of $\Gamma_n^+$. Label all corresponding points along $\Gamma_n$ as $\eta_{-s}, \eta_{-s+1}, \cdots, \eta_{-1}, \eta_{0}$.

Define $\mathcal{M}_n$ by labeling the intersections of\, $\Gamma_n$ and the real axis as $\xi_n$, and then
\be
\label{mn}
    \mathcal{M}_n(z)=\begin{cases}
        \sum\limits_{s=n+1}^{-1}\left(\mathcal{M}_{n+1}^+(\zeta_s)+\mathcal{M}_{n+1}^-(\zeta_s)\right), & z\in[\zeta_n,\zeta_{n+1}], \qquad n\leq-2\\
        0, & z\in[\zeta_{n},\zeta_{n+1}], \qquad n=-1\\
        \sum\limits_{s=0}^{n}\left(\mathcal{M}_n^+(\zeta_s)+\mathcal{M}_n^-(\zeta_s)\right), & z\in[\zeta_n,\zeta_{n+1}], \qquad n\geq0
    \end{cases}.
\ee
%
%
	To explicitly express the jump contour for the RH problem, we define $\tilde{\Gamma}$ as
	\be\label{e:~Gamma}
	\~\Gamma = \Gamma\setminus\big(\bigcup_{n=G_-}^{G_+} \Gamma_{n,0}^{\pm}\cup \Gamma_\mathrm{res}\big),
	\ee
	where $\Gamma$ is defined in \eqref{e:Gamma}.
\end{definition}

\noindent
The above choices of branch cuts imply 

\begin{proposition}
With the above choices of branch cut, 
$r_{1/4}(z) = \i^n|\D(z)|^{1/2}(1 + O(1/z))$ as $|\Im z|\to\infty$ along $\Gamma_n$. 
\label{p:r14asymp}
\end{proposition}
\begin{proof}
Let $z=x+\i y$. From Proposition \ref{asymD} we have that $\D(z)\sim \cos(zL)$ as $|y|\to\infty$,
where as usual ``$\sim$'' denotes equality up to higher-order terms. 
Therefore, as $|y|\to\infty$,
$r_{1/4}(z)\sim \cos^{1/2}(zL) = (\cos xL \cosh yL + \i \sin xL \sinh yL)^{1/2}
  \sim \cosh^{1/2}(yL)\,\e^{\i xL/2}$.
This completes the proof.
\end{proof}
\noindent Most importantly, the above definitions allow us to obtain the main result of this section, namely: 
\begin{theorem}
	The $2\times2$ matrix-valued function $\Phi(x,z)$ defined by~\eqref{Phi} solves the following Riemann-Hilbert problem:
	\label{t:direct}
	\begin{RHP}
		\label{RHP1}
		Find a $2\times2$ matrix-valued function $\Phi(x,z)$ such that
		\vspace*{-1ex}
		\begin{enumerate}
			\advance\itemsep-2pt
			\item 
			$\Phi(x,z)$ is a holomorphic function of $z$ for $z\in\C\setminus\~\Gamma$. 
			\item 
			The non-tangential limits $\Phi_\pm(x,z)$ of $\Phi(x,z)$ to $\~\Gamma$ are continuous functions of $z$ in $\~\Gamma\setminus\{E_k\}$, and have at worst quartic root singularities on $\{E_k\}$.
			\item 
			$\Phi_\pm(x,z)$ satisfy the jump relation 
			\be
			\Phi_+(x,z)=\Phi_-(x,z)V(x,z),\qquad z\in \~\Gamma \,,
			\ee
			where $V(x,z)$ is given by~\eqref{e:Vexplicit}, and $\~\Gamma$ is given by \eqref{e:~Gamma}.
			\item 
			As $z\to\infty$, $\Phi(x,z)$ has the following asymptotic behavior 
			\be
			\Phi(x,z) = (I+O(1/z))\,B(z), 
			\label{asymPhi}
			\ee
			with $U$ as in~\eqref{e:Udef} and $B(z)$ as in~\eqref{B}.
			\item 
			There exist positive constants $c$ and $M$ such that $|\phi_{ij}(x,z)|\leq M\e^{c|z|^2}$ for all $z\in \mathcal{D}$.
		\end{enumerate}
	\end{RHP}
\end{theorem}
\begin{proof} 
We prove each item separately.
	
\noindent\textit{Item~1}. 
It follows from \eqref{BF} that $\psi^\pm$ are meromorphic functions of $z$ for $z\in\mathbb{C}\setminus\Sigma(\L)$, so that $\Psi(x,z)$ is meromorphic in $\mathbb{C}\setminus\Sigma(\L)$. From \eqref{B}, we can derive that $\Phi(x,z)$ could only be singular on $\{\mu_j\}$ and $\{\z_j\}$. Then $\Phi(x,z)$ is a holomorphic function of $z$ for $z\in\C\setminus\~\Gamma$. 

\noindent\textit{Item~2}.
\label{cond2}
    From the definition of the Bloch-Floquet solutions \eqref{BF}, we can derive that $\Psi_{\pm}(x,z)$ can only be singular at the Dirichlet eigenvalues $\mu_j$. From \eqref{B} and \eqref{Phi}, we also get that $\Phi_{\pm}(x,z)$ can only be singular at $\mu_j$ and $\zeta_j$. We will now prove that $\Phi_{\pm}$ cannot be singular at $\mu_j$ unless $\mu_j=E_k$ for some $k$. We will discuss three different cases: $\nu_j=\pm1$ and $\nu_j=0$. We will first examine the behavior at $z=\mu_j$ then use the two to one map between $\zeta$ and $\mu$ to examine the behavior at $z=\zeta_{j}$.

    Suppose that $\nu_j=1$, which is only possible when $\mu_j=\mucirc_k$ for some $k$. By looking at the definition of $\nu_j$, we can derive that $\~y_{11}(L,\mu_j)=\rho^{-1}(\mu_j)$. Then $\rho^{-1}(z)-\~y_{11}(L,z)$ is holomorphic in a neighborhood of $\mu_j$, and $\mu_j$ are the zeros of $\rho^{-1}(z)-\~y_{11}(L,z)$. For $z\in\mathbb{C}^+$, we have
\be
\Phi(x,z) = i b(z) \,\big(\,\psi^-f^-\,,\,\psi^+f^+\,\big)\,\e^{izx\sigma_3}\,.
\ee

    For the first column, the zeros of $\rho^{-1}(z)-\~y_{11}(L,z)$ at $\mu_j$ cancel the zeros of $\~y_{12}$ at $\mu_j$. For the second column, the zeros of $f^+$ cancel the singularities of $\psi^+$. Therefore $\Psi_+$ is nonsingular at $\mu_j$.  We reach the same conclusion for $z\in\mathbb{C}^-$ by the same method.

    Next, since $\mu_j=\mucirc_k$ for some $k$, the two to one map from $\zeta$ to $\mu$ means that the two corresponding $\zeta_j$ must be $E_{2k}$ and $E_{2k-1}$. Once again looking at \eqref{B}, we see that the only singular contribution to $\Phi_{\pm}$ is at $z\rightarrow E_{2k}$ and $E_{2k-1}$ from the boundary values of 
    $r_{1/4}(z)$. 

    Now suppose $\nu_j=-1$, which is once again only possible when $\mu_j=\mucirc_k$ for some $k$. 
    We have $\~y_{11}(L,\mu_j)=\rho(\mu_j)$. Then $\rho(z)-\~y_{11}(L,z)$ is holomorphic in $\mu_j$'s neighborhood, and $\mu_j$ are the zeros of $\rho(z)-\~y_{11}(L,z).$ 
    Using the same method as for $\nu=1$, we see that $\Phi_{\pm}$ are nonsingular at $\mu_j$ and has quartic root singularities at $E_{2k}$ and $E_{2k-1}$

    Lastly suppose that $\nu_j=0$. This can happen in two scenarios: a Dirichlet eigenvalue is in a degenerate gap $(\mu_j=\zeta_{2j}=\zeta_{2j-1})$ or a Dirichlet eigenvalue is on a band edge $(\mu_j=\mucirc_k=E_{2k}$ or $\mu_j=\mucirc_k=E_{2k-1}$ for some $k$). We examine the former first.

    If $\nu_j=0$ then $\rho(\mu_j)=\rho^{-1}(\mu_j)=\pm1$. 
    Therefore $\mu_j$ are the zeros of both $\rho(z)-\~y_{11}(L,z)$ and $\rho^{-1}(z)-\~y_{11}(L,z)$. 
    Looking at $\Psi_{\pm}$, the zeros of $\rho(z)-\~y_{11}(L,z)$ and $\rho^{-1}(z)-\~y_{11}(L,z)$ at $\mu_j$ cancel the zeros of $\~y_{12}(L,z)$ at $\mu_j$. 
    Furthermore, when $\mu_j=\zeta_{2j}=\zeta_{2j-1}$, the square root zeros of $(f^0)^{1/2}$ at $\mu_j$ cancel those of  $r_{1/4}(z)$ at $\zeta_{2j}=\zeta_{2j-1}$. 
    In this case, $\Psi_{\pm}(x,z)$ is nonsingular at $z=\mu_j$. 

    Now let $\mu_j=\mucirc_k=E_{2k}$ for some $k$. In this case we still have $\nu_j=0$ and $\Psi_{\pm}$ is still nonsingular at $\mu_j$. 
    Examining the coefficient of the matrix $B(z)$ in \eqref{B}, we see that $(f^0)^{1/2}$ has a square root zero and $r_{1/4}(z)$ has a quartic root zero at $\mu_j$. This leads to a quartic root zero in the numerator which is not a singularity. 
    However, at $E_{2k-1}$, there is no zero in $(f^0)^{1/2}$ meaning there is a quartic root singularity at $z=E_{2k-1}$ when $\mu_j=E_{2k}$. Likewise, there is a quartic root singularity at $z=E_{2k}$ when $\mu_j=E_{2k-1}$.
 
\noindent\textit{Item~3}. 
Recalling the definition of $\rho(z)$ in section~\ref{s:spectrum},
	we have that 
	the discrepancy between the boundary values of $\sqrt{\D^2-1}$ on the branch cut~$\Sigma(\L)$ is 
	\vspace*{-0.7ex}
	\bse
	\begin{gather}
	\rho_+(z)=\rho_-^{-1}(z),\qquad \rho_+^{-1}(z)=\rho_-(z),\qquad  z\in\Sigma(\L).
	\end{gather}
	\ese
	Therefore, 
	\vspace*{-0.4ex}
	\bse
	\begin{gather}
		\Psi_+(x,z)=\Psi_-(x,z)\sigma_1\,,\quad z\in\Sigma(\L)\setminus\Real\,.
		\\
		\Psi_+(x,z)=\Psi_-(x,z)\,,\quad z\in\Real\,.
	\end{gather}
   \ese
	As a result, we see that the jump matrix $V(x,z)$ takes the form 
	\begin{align*}
	V(x,z) &= \e^{-\i zx\sigma_3}(B_-)^{-1}\sigma_1B_+\e^{\i zx\sigma_3}, \quad z\in\Sigma(\L)\setminus\Real,\\
	&=\begin{cases}
	& \displaystyle \frac{h_-(z)}{h_+(z)}\frac{p_+(z)}{p_-(z)}\e^{-\i zx\hat\sigma_3}
	\begin{pmatrix}
		0 & {f^+}/{f^-}\\
		{f^-}/{f^+} & 0
	\end{pmatrix}, \quad z\in\Complex^+\cap \Sigma(\L),\\
    & \displaystyle \frac{h_-(z)}{h_+(z)}\frac{p_+(z)}{p_-(z)}\e^{-\i zx\hat\sigma_3}
    \begin{pmatrix} 0 & {f^-}/{f^+}\\ {f^+}/{f^-} & 0 \end{pmatrix}, \quad z\in\Complex^-\cap \Sigma(\L),
	\end{cases}\\
\begin{split}
V(x,z)&=\e^{-\i zx\sigma_3}(B_-)^{-1}B_+\e^{\i zx\sigma_3}, \quad z\in\Real,\\
&=\i\frac{h_-(z)}{h_+(z)}\frac{p_+(z)}{p_-(z)}\e^{-\i zx\hat\sigma_3}
\begin{pmatrix}
	{f^-}/{f^+} & 0\\
	0 & {f^+}/{f^-}
\end{pmatrix}, \quad z\in\Real.
\end{split}
	\end{align*}
	where $p(z)$ and $h(z)$ are defined as
		\bse
	\begin{gather}
		p(z) = \left(\prod_{k=\g^-\atop\nu_k=0 }^{\g^+}(\mucirc_k - z)\,\right)^{\frac{1}{2}}\,,
		\label{e:pdef}
		\\
		h(z) = \left(\prod_{k=\g^- }^{\g^+}{(E_{2k-1}-z)(E_{2k}-z)}
		\right)^{\frac{1}{4}}\,.
	\end{gather}
	\ese
On the real axis, we have
	\be
	\frac{h_-(z)}{h_+(z)}=-\i,\quad z\in[\varpi_{2k},\varpi_{2k+1}],\qquad \frac{h_-(z)}{h_+(z)}=\i,\quad z\in[\varpi_{2k+1},\varpi_{2k+2}].
	\ee
On the analytic curve $\Gamma_n$, we have
	\be
\frac{h_-(z)}{h_+(z)}=(\mp\i)^k,\quad z\in \Gamma_{n,k}^\pm.
\ee
Putting everything together, we obtain the jump matrix explicitly as in~\eqref{e:Vexplicit}.

\noindent\textit{Item~4}. 
From the asymptotic behaviors of $\psi^{\pm}(x,z)$ in Proposition~\ref{p:asympsi}, we have 
    \be
        \Psi(x,z)= \left(\sqrt2  U+O(1/z)\right)\e^{-izx\sigma_3},\qquad z\to\infty\,.
    \ee
Thus, recalling~\eqref{e:Utildedef}, we have 
    \be
        \~U \, \Psi(x,z)\e^{izx\sigma_3}=  I+O(1/z), \qquad z\rightarrow\infty 
    \ee
for $z\in\mathbb{C}$, which implies~\eqref{asymPhi}

\noindent\textit{Item~5}. 
The proof of item~5 is the same as the proof for the defocusing NLS equation, and is therefore omitted for brevity. 
Details can be found in \cite{BiondiniZhang2023}.
\end{proof}

\smallskip
\begin{lemma}
\label{equiv1}
    If  $\Phi(x,z)$ solves the RHP~\ref{RHP1}, then $\det\Phi(x,z)\equiv1$.
\end{lemma}

\begin{theorem}
\label{unique}
For any fixed $x\in\mathbb{R}$, the solution of the RHP~\eqref{RHP1} is unique.
\end{theorem}

\begin{corollary}
        The potential matrix $Q(x)$ of the Dirac equation is obtained from the solution of any solution $\Phi(x,z)$ of RHP \eqref{RHP1} by
        \be
        \label{e:reconstruction2} 
        Q(x) = \lim_{z\to\infty} iz\left[\sigma_3,\Phi(x,z)B^{-1}(z)\right]\,. 
        \ee
\end{corollary}

The reconstruction formula for $Q(x)$ in \eqref{e:reconstruction2} follows from \eqref{Phi} and Corollary~\ref{e:reconstruction}.
The proofs of Lemma~\eqref{equiv1} and Theorem~\eqref{unique} are once again the same as in the defocusing case and are omitted for brevity. 
Details can be found in~\cite{BiondiniZhang2023}.

A feature in the non-self-adjoint (focusing) case that is not present in the self-adjoint (defocusing) case  
is the possibility of a Dirichlet eigenvalue lying on the interior of a spectral band. 
(Note that, due to the symmetry \eqref{symmetry} this phenomenon can only happen off the real axis,
see Appendix~\ref{a:spectrum}.)
This situation is already accounted for by the formalism, as briefly discussed next.

To see what happens if a Dirichlet eigenvalue lies inside a band, it suffices to look at the definition~\ref{nudef} of $\nu_k$ as well as that of the Bloch Floquet solutions \eqref{BF}. 
When there is a Dirichlet eigenvalue $\mu$ on the interior of a band, then $\~y_{11}(L,\mu)$ and $\~y_{22}(L,\mu)$ are complex conjugates  with magnitude 1. 
By \eqref{nudef}, this implies that $\nu=0$ for this Dirichlet eigenvalue.
For all other Dirichlet eigenvalues, $\nu=0$ implies that $\psi^\pm$ have no pole there, by Proposition~\ref{p:psipm_nu=0}.
However, when $\mu$ lies in the interior of a band, 
$\~y_{11}(L,\mu)$ and $\~y_{22}(L,\mu)$ are not equal to $\pm1$, per the above considerations,
implying that there is indeed a pole for either $\psi^+$ or $\psi^-$. 
The resolution to this complication is to first identify whether $\~y_{11}(L,\mu)$ or $\~y_{22}(L,\mu)$ equals 
$\rho(\mu)$, which determines whether $\psi^+(x,z)$ or $\psi^-(x,z)$ has a pole at $z=\mu$. 
From there, one then removes this Dirichlet eigenvalue in $f^0(z)$, instead placing it in either $f^+(z)$ or $f^-(z)$, 
in order to eliminate the pole from the Riemann-Hilbert problem.\break
A concrete example of this situation is discussed in Appendix~\ref{s:genuszero}.

\section{Time dependence, initial value problem for the focusing NLS equation with periodic boundary conditions}
\label{s:time}

In this section we show how the results of sections~\ref{s:spectrum}, \ref{s:asymptotics} and~\ref{s:inverse} can be used to 
solve the initial value problem for the focusing NLS equation with periodic BC, namely
\vspace*{-0.4ex}
\be
q(x+L,t)=q(x,t)\,.
\ee

To construct a Riemann-Hilbert characterization of periodic solutions of the focusing NLS equation with infinite-gap initial conditions, we recall 
that the Lax equation of the focusing NLS equation~\eqref{e:nls}
associated with the Lax pair~\eqref{e:NLSLP},
which can be derived from the zero curvature condition $X_t - T_x + [X,T] = 0$ in a straightforward way, is
\vspace*{-0.4ex}
\bse
\begin{gather}
	\L_t=[\mathcal{A},\L]\,,
    \label{e:LaxequationNLS}
    \\[-0.4ex]
	\intertext{with $\L$ as in~\eqref{e:Diracoperator} and}
	\mathcal{A}=2\i\sigma_3\partial_x^2+2\i Q\sigma_3\partial_x-\i(Q^2-Q_x)\sigma_3\,.
\end{gather}
\ese
It is well known that, if the potential $q$ of~\eqref{e:Diraceigenvalueproblem} evolves according to the NLS equation~\eqref{e:nls}, the Lax spectrum $\Sigma(\L)$ of the ZS problem is invariant in time
(which is easily proven by showing that the Floquet discriminant $\Delta(z)$ is time-independent).
On the other hand, the movable Dirichlet eigenvalues are time dependent, and so are the Bloch-Floquet eigenfunctions.
Thus, in order to construct an effective time-dependent Riemann-Hilbert problem, 
one must determine the time evolution of the Dirichlet eigenvalues and 
of the Bloch-Floquet eigenfunctions.
We turn to these tasks next.

\begin{proposition}
	For all $n\in\Integer$, the Dirichlet eigenvalues $\mu_n(t)$ satisfy the ODE
	\bse
	\label{e:dmundt}
	\be
	\partialderiv{\mu_{n}(t)}t=\nu_n(t)c_1(\mu_{n}(t))\frac{\rho(\mu_n(t))-\rho^{-1}(\mu_n(t))}{\~y'_{12}(L,t,\mu_{n}(t))}.
	\ee
	where 
	\be
	c_1(z)=-2z^2+\i z(q^*(0,t)+q(0,t))+|q(0,t)|^2+\half(q_x^*(0,t)-q_x(0,t))\,.
	\label{e:c1def}
	\ee
	\ese
	\label{p:dirichlettimeevol}
\end{proposition}
\begin{proof}
	Recall that the Dirichlet eigenvalues are poles of the modified Bloch-Floquet solutions,
	which, in turn, are defined via~\eqref{e:Ytildedef}.
	The transformation~\eqref{e:Ytildedef} yields the modified Lax equation
    \vspace*{-0.4ex}
	\bse
	\begin{gather}
		\tilde{\L}_t=[\tilde{\mathcal{A}},\tilde{\L}]\,,
		\label{mlax}
		\\[-1ex]
		\intertext{with}
		\tilde{\L}= U\L U^{-1}\,,
		\qquad
		\tilde{\mathcal{A}} = U\mathcal{A}U^{-1}\,.
	\end{gather}
	\ese
	Recall that $\~y_2$ is a Bloch-Floquet solution of $\tilde\L$, 
	i.e., $\tilde{\L}\~y_2=z\~y_2$. 
	Differentiating both sides with respect to $t$, we have
	$\tilde{\L}_t\~y_2+\tilde{\L}\~y_{2,t}=z\~y_{2,t}$,
	which, 
	combined with the modified Lax equation \eqref{mlax} implies 
	\be
	\tilde{\L}(\~y_{2,t}-\tilde{\mathcal{A}}\~y_2)=z(\~y_{2,t}-\tilde{\mathcal{A}}\~y_2)\,.
	\ee
	Therefore, for fixed $t$ and $z$, there exist constants $c_1$ and $c_2$ such that
	\be
	\~y_{2,t}-\tilde{\mathcal{A}}\~y_2=c_1\~y_1+c_2\~y_2\,.\label{y2t}
	\ee
	The normalization condition of $\~Y(x,z)$ at $x=0$ gives that $\~y_{12,t}(0,t,z)=\~y_{22,t}(0,t,z)=0$, so we obtain
	\be
	(c_1,c_2)^T = -\tilde{\mathcal{A}}\~y_{2}(0,t,z)\,.
	\ee
	Moreover, the modified equation \eqref{e:modZS} also implies
	\be
	\~y_{2,xx}(x,t,z)=-\i zU\sigma_3U^{-1}\~y_{2,x}(x,t,z)+UQ(x,t)U^{-1}\~y_{2,x}(x,t,z)+UQ_x(x,t)U^{-1}\~y_2(x,t,z)\,,
	\ee
	and therefore
	\begin{align}
			\tilde{\mathcal{A}}\~y_2(x,t,z) 
			= -2\i z^2U\sigma_3U^{-1}\~y_2(x,t,z)+2zUQU^{-1}\~y_2(x,t,z)-\i U(Q^2+Q_x)\sigma_3U^{-1}\~y_2(x,t,z)\,.
		\label{Ay2}
	\end{align}
	Evaluation at $x=0$ yields
	\bse
	\label{c1c2}
	\begin{align}
		&c_1=-2z^2+\i z(q^*(0,t)+q(0,t))+|q(0,t)|^2+\half(q_x^*(0,t)-q_x(0,t))\,,\label{c1}\\
		&c_2=z(q(0,t)-q^*(0,t))+\frac{\i}{2}(q_x^*(0,t)-q_x(0,t))\,.\label{c2}
	\end{align}
	\ese
Inserting~\eqref{Ay2} and~\eqref{c1c2} into \eqref{y2t} and evaluating~\eqref{y2t} at $x=L$ and $z=\mu_n(t)$, we get 
	\be
		\~y_{12,t}(L,t,\mu_n(t)) = c_1(\mu_n(t))
		\big(\~y_{11}(L,t,\mu_n(t))-\~y_{22}(L,t,\mu_n(t))\big)\,.
	\ee
At the same time, differentiating the expansion \eqref{y12} in $t$ we have
\vspace*{-0.4ex}
\be
\~y_{12,t}(L,t,z)=\e^{A_1z+B_1}\sum_{n\in\mathbb{Z}}\left[\frac{z\mu_{n,t}(t)}{\mu_n^2(t)}\e^{\frac{z}{\mu_n(t)}}\frac{z}{\mu_n(t)}\prod_{j\neq n}\left(1-\frac{z}{\mu_j(t)}\right)\e^{\frac{z}{\mu_j(t)}}\right]\,,
\label{y12t}
\ee
whereas differentiating  \eqref{y12} in $z$ we have
\be
\~y'_{12}(L,t,z)=A_1\e^{A_1z+B_1}\prod_{j\in\mathbb{Z}}\left(1-\frac{z}{\mu_j(t)}\right)\e^{\frac{z}{\mu_j(t)}}+\e^{A_1z+B_1}\sum_{n\in\mathbb{Z}}\left[-\frac{z}{\mu_n^2(t)}\e^{\frac{z}{\mu_n(t)}}\prod_{j\neq n}\left(1-\frac{z}{\mu_j(t)}\right)\e^{\frac{z}{\mu_j(t)}}\right]\,.
\label{y12z}
\ee
Evaluating \eqref{y12t} and \eqref{y12z} at $z=\mu_n(t)$ gives $\~y_{12,t}(L,t,\mu_n(t))=-\mu_{n,t}(t)\~y'_{12}(L,t,\mu_{n}(t))$.
Therefore, we can express $\mu_{n,t}(t)$ as:
\be
\mu_{n,t}(t) = -\frac{\~y_{12,t}(L,t,\mu_n(t))}{\~y'_{12}(L,t,\mu_{n}(t))}
= \frac{c_1(\mu_{n}(t))(\~y_{22}(L,t,\mu_n(t))-\~y_{11}(L,t,\mu_n(t)))}{\~y'_{12}(L,t,\mu_{n}(t))}\,,
\nonumber
\ee
which yields \eqref{e:dmundt} upon substituting $c_1$ via~\eqref{e:c1def}.
\end{proof}

\begin{proposition}
	\label{asymofa}
	The time-dependent Bloch-Floquet solutions $\psi^\pm(x,t,z)$ satisfy the ODEs
	\be
	\psi^\pm_t(x,t,z)+\alpha^\pm(t,z)\psi^\pm(x,t,z)=\tilde{\mathcal{A}}\psi^\pm(x,t,z)
	\label{psit}
	\ee
	where 
	\vspace*{-1ex}
	\begin{align}
		\begin{split}
			\alpha^\pm(t,z)=&z(q(0,t)-q^*(0,t))+\frac{\i}{2}(q_x(0,t)+q^*_x(0,t))\\
			&+\Big(2z^2-\i z\big(q^*(0,t)+q(0,t)\big)-|q(0,t)|^2+\half q_x(0,t)-\half q_x^*(0,t)\Big)
			\psi^\pm_2(0,t,z)\,.
		\end{split}
		\label{e:alphapm}
	\end{align}
	As a result, as $z\rightarrow\infty$ we have 
	\be
	\label{e:alphapmasymp}
		\alpha^\pm= \begin{cases} \pm2\i z^2+O(1/z),\qquad z\in\Complex^+,\\
		\mp2\i z^2+O(1/z),\qquad z\in\Complex^-.
	\end{cases}
	\ee
	Furthermore, if $\nu_n(t)=1$ then $\alpha^+(t,z)$ has a simple pole at $\mu_{n}(t)$ with residue $-\mu_{n,t}(t)$ and if $\nu_{n}=-1$ then $\alpha^-(t,z)$ has a simple pole at $\mu_{n}(t)$ with residue $-\mu_{n,t}(t)$.
\end{proposition}

\begin{proof}
	For each $t$, we can factor the normalized Bloch-Floquet solutions $\psi^\pm(x,t,z)$  as
	\be
	\psi^+(x,t,z)=p^+(x,t,z)\rho^{x/L}(z),\qquad 
	\psi^-(x,t,z)=p^-(x,t,z)\rho^{-x/L}(z)\,.
	\ee
	Differentiating $(\tilde{\L}-z)\psi^+=0$ with respect to $t$ and using the modified Lax equation imply that
	\be
	\tilde{\L}_t\psi^++\tilde{\L}\psi^+_t-z\psi^+_t=[\tilde{\mathcal{A}},\tilde{\L}]\psi^++\tilde{\L}\psi^+_t-z\psi_t^+=(\tilde{\L}-z)(\psi^+_t-\tilde{\mathcal{A}}\psi^+)=0.
	\ee
	So $\psi_t^+-\tilde{\mathcal{A}}\psi^+$ solves~\eqref{e:modZS} for all $t>0$. 
	Therefore, it can be written as a linear combination of $\psi^\pm$:
	\be
	p^+_t\rho^{x/L}(z) - \tilde{\mathcal{A}}\,p^+\rho^{x/L}(z) = 
	- \alpha^+(t,z)p^+\rho^{x/L}(z) - \beta^+(t,z)p^-\rho^{-x/L}(z),
	\ee
	where $\alpha^+(t,z)$ and $\beta^+(t,z)$ are some undetermined functions independent of $x$. 
	Recall that $|\rho(z)|<1~\forall z\in\mathbb{C}\setminus\Sigma(\L)$.
	Evaluating the LHS and RHS of the above equation as $x\to\infty$,
	every term decays exponentially except the last one.
	Therefore, in order for the equation to hold, one needs $\beta^+(t,z)=0$. 
	Hence $\psi^+$ solves~\eqref{psit} with the plus sign.
	An analogous argument, but now taking the limit as $x\to-\infty$, 
	shows that $\psi^-$ solves~\eqref{psit} with the minus sign, for some $\alpha^-(t,z)$. 
	Next we derive~\eqref{e:alphapm} and~\eqref{e:alphapmasymp}. 
	To do so, we evaluate the first component of~\eqref{psit} at $x=0$.
	Since $\psi^\pm_1(0,t,z)=1$ for all time, we have $\psi^\pm_{1,t}(0,t,z)=0$. 
	Therefore \eqref{psit} yields
	\be
	\alpha^\pm(t,z) = (\~{\mathcal{A}}\psi^\pm)_1(0,t,z)\,.
	\label{e:alphapmeq}
	\ee
	From \eqref{Ay2} we have that
	\be
	\tilde{\mathcal{A}}v = -2\i z^2U\sigma_3U^{-1}v+2zUQU^{-1}v-\i U(Q^2+Q_x)\sigma_3U^{-1}v\,.
	\ee
	Explicitly, the first component of the above equation is
	\be
	(\tilde{\mathcal{A}}v)_1 = 2z^2\,v_2 + z\,\big[(-q^* + q)\,v_1 - \i(q^*+q)v_2 \big]
	\\
	+ \txtfrac \i2 \big[ (q^*_x + q_x)\,v_1 + \i(2|q|^2 + q^*_x - q_x ) v_2 \big]\,.
	\label{e:alphapmeq2}
	\ee
	Letting $v(x,t,z) = \psi^\pm(x,t,z)$, evaluating \eqref{e:alphapmeq2} at $x=0$ 
	recalling that~\eqref{BF} evaluated at $x=0$ yields
	$\psi^\pm_1(0,t,z)=1$ and 
	$\psi^\pm_2(0,t,z) = (\rho^{\pm1}-\~y_{11}(L,t,z))/\~y_{12}(L,t,z)$,
	\eqref{e:alphapmeq} then yields~\eqref{e:alphapm}.
\end{proof}

Next, we will introduce an auxiliary function $e^\pm(t,z)$ to set up the time-dependent Bloch-Floquet solutions.
Let $e^\pm(t,z)$ be solutions of the differential equation  
$e^\pm_t(t,z)=\alpha^\pm(t,z)e^\pm(t,z)$ with the initial condition $e^\pm(0,z)=1$,
i.e.,
\vspace*{-1ex}
\be
e^\pm(t,z)=\exp\left(\int_0^t\alpha^\pm(\tau,z)\d\tau\right)\,.
\ee
The functions $e^\pm(t,z)$ satisfy the following properties, the proof of which is essentially identical to that in the defocusing case, and can be found in \cite{BiondiniZhang2023}:

\begin{proposition}
	\label{p:e}
	The functions $e^\pm(t,z)$ satisfy the following properties:
	\vspace*{-1ex}
	\begin{enumerate}
		\advance\itemsep-2pt
		\item 
		$e^\pm(t,z)$ are meromorphic functions in $\mathbb{C}\setminus\Sigma(\mathcal{L})$.
		\item 
		$e^+(t,z)$ has simple poles on $\mathring\mu_{k}(0)$ as $\nu_{k}(0)=1$and simple zeros on $\mathring\mu_{k}(t)$ as $\nu_{k}(t)=1$;
		$e^-(t,z)$ has simple poles on $\mathring\mu_{k}(0)$ as $\nu_{k}(0)=-1$ and simple zeros on $\mathring\mu_{k}(t)$ as $\nu_{k}(t)=-1$. 
		\item
		$e^\pm(t,z)$ both have square root singularities at $\mathring\mu_{k}(0)$ when $\nu_{k}(0)=0$ and square root zeros at $\mathring\mu_{k}(t)$ when $\nu_{k}(t)=0$.
		\item 
		The boundary values $e_\pm^\pm$ of $e^\pm$ satisfy $e^\pm_+(t,z)=e^\mp_-(t,z)$ for $z\in\Sigma(\mathcal{L})$.
		\item 
		For fixed $t$, $e^\pm$ have the following asymptotic behaviors 
		as $z\rightarrow\infty$:
		\vspace*{-0.6ex}
		\be
		e^\pm(t,z)=\begin{cases} \e^{\pm2\i z^2t}(1+O(1/z)),& z\in\Complex^+\\
			\e^{\mp2\i z^2t}(1+O(1/z)),& z\in\Complex^-\end{cases}\,.
		\ee
		\vspace*{-2ex}
		\item 
		There exist positive constants $c$ and $M$ such that $|e^\pm(t,z)|\leq M\e^{c|z|^2}$.
	\end{enumerate}
\end{proposition}

Thanks to the properties of $e^\pm(t,z)$, we proceed to introduce the time-dependent solutions and the time-dependent jump matrix for the RH problem.
Define
\be\label{checkpsi}
\check{\psi}^\pm(x,t,z)=\psi^\pm(x,t,z)e^\pm(t,z)\,.
\ee
Then $\check{\psi}^\pm(x,t,z)$ satisfy the system of equations $\~\L\check\psi^\pm = z\check\psi^\pm$ and 
$\check{\psi}^\pm_t=\tilde{\mathcal{A}}\check{\psi}^\pm$,
for which the modified Lax equation~\eqref{mlax} serves as  the compatibility condition. 
Note that the factors $e^\pm(t,z)$ in Eq.~\eqref{checkpsi} have the effect of shifting the singularities of $\check \psi^\pm(x,t,z)$ from those of
$\psi^\pm(x,t,z)$ to those of the Bloch-Floquet eigenfunctions at time zero, $\psi^\pm(x,0,z)$.

\begin{definition}
	Let $\check{V}(x,t,z)$ be defined by
	\be
	\check{V}(x,t,z) =\begin{cases}
			(-1)^{\mathcal{M}_n(z)+1}\begin{pmatrix} {f^-}/{f^+} & 0\\ 0 & {f^+}/{f^-} \end{pmatrix},
			& z\in\Omega\,,\\
			(-1)^{\mathcal{M}_n(z)}\begin{pmatrix} {f^-}/{f^+} & 0\\ 0 & {f^+}/{f^-} \end{pmatrix},
			&z\in\Real\setminus\Omega\,,\\
			(-1)^{\mathcal{M}_n^+(z)}(-\i)^k\e^{-\i( zx+2z^2t)\hat\sigma_3}
			\begin{pmatrix}	0 & {f^+}/{f^-}\\ {f^-}/{f^+} & 0\end{pmatrix}, & 
            z\in \Gamma_{n,k}^+,\quad $k$~\textrm{odd},\\
			(-1)^{\mathcal{M}_n^+(z)}(-\i)^k, & 
            z\in \Gamma_{n,k}^+,\quad $k$~\textrm{even},\\
			(-1)^{\mathcal{M}_n^-(z)}\i^k\e^{-\i( zx+2z^2t)\hat\sigma_3}
			\begin{pmatrix}	0 & {f^-}/{f^+} \\ {f^+}/{f^-} & 0 \end{pmatrix}, &  
            z\in\Gamma_{n,k}^-,\quad $k$~\textrm{odd},\\
			(-1)^{\mathcal{M}_n^-(z)}\i^k, & z\in \Gamma_{n,k}^-,\quad $k$~\textrm{even},
		\end{cases}	
	\label{e:Vxtz}
	\ee
	with $\mathcal{M}_n(z)$ and $\mathcal{M}_n^\pm(z)$ still the counting functions \eqref{mnplus}-\eqref{mn},
	and let the matrix-valued time-dependent Bloch-Floquet solution $\check\Psi(x,t,z)$ be defined as in~\eqref{e:defPsi}
	but with $\psi_j^\pm$ replaced by $\check\psi_j^\pm$, as defined in \eqref{checkpsi}.
	Finally, let $\check{\Phi}(x,t,z)$ be given by
	\be
	\label{Phicheck}
        \check{\Phi}(x,t,z)= (2U)^{-1}\, \check{\Psi}(x,t,z)B(z)\e^{\i zx\sigma_3+2\i z^2t\sigma_3}\,.
\ee
\end{definition}

\begin{theorem}
	\label{qtoPhi}
	Let $q(x,t)$ be the solution to the focusing NLS equation~\eqref{e:nls}
	with smooth initial data $q(x,0)=q_0(x)$, 
	and let $\Sigma(q_0)$ denote the spectral data corresponding to the Dirac operator associated with $q_0(x)$.
	Then, there exists a unique solution $\check{\Phi}$ to the following Riemann-Hilbert problem:
	\begin{RHP}
		\label{RHP2}
		Find a $2\times2$ matrix-valued function $\check{\Phi}(x,t,z)$ such that
		\vspace*{-1ex}
		\begin{enumerate}
			\advance\itemsep-4pt
			\item 
			$\check{\Phi}(x,t,z)$ is a holomorphic function of $z$ for $z\in\C\setminus\~\Gamma$.
			\item 
			The non-tangential limits $\check{\Phi}_\pm(x,t,z)$ are continuous functions of $z$ in $\begingroup\~\Gamma\endgroup\setminus\{E_k\}$, and have at worst quartic root singularities on $\{E_k\}$.
			\item 
			$\check{\Phi}_\pm(x,t,z)$ satisfy the jump relation $\check{\Phi}_+(x,t,z)=\check{\Phi}_-(x,t,z)\check{V}(x,t,z)$, with $\check{V}(x,t,z)$ given by \eqref{e:Vxtz}.
			\item 
			As $z\to\infty$ with $z\in\Complex$, $\check{\Phi}(x,t,z)$ has the following asymptotic behavior 
			\be
			\check{\Phi}(x,t,z) = \,(I+O(1/z))\,B(z)\,.  \label{asymcheckPhi}
			\ee
			\item 
			There exist positive constants $c$ and $M$ such that $|\check{\phi}_{ij}(x,t,z)|\leq M\e^{c|z|^2}$ for all $z\in \mathcal{D}$.
		\end{enumerate}
	\end{RHP}
\end{theorem}

For fixed $x,t\in\mathbb{R}$, the solution of  RHP \ref{RHP2} is unique.
Moreover, the matrix $Q(x,t)$ associated with the solution $q(x,t)$ of the NLS equation~\eqref{e:nls}
is given in terms of the solution $\check{\Phi}(x,t,z)$ of RHP \ref{RHP2} by 
\be
Q(x,t) = \lim_{z\to\infty}\i z [\sigma_3,\check{\Phi}(x,t,z)B^{-1}(z)]\,.
\label{trecons}
\ee

\section{Concluding remarks}
\label{s:conclusions}

In summary, we have presented the formulation of the direct and inverse spectral theory for a non-self-adjoint Dirac operator via a Riemann-Hilbert approach that remains valid in the infinite genus case, and we showed how the formalism can be used to solve the initial value problem for the focusing NLS equation.

Compared to the solution representation obtained via the trace formulae (e.g., see \cite{McLaughlinOverman}), 
the present formalism allows one to recover the solution by specifying the value of a single set of Dirichlet eigenvalues
(as opposed to two) at a single base point.
Also, compared to both the trace formulae and the solution representation obtained via the algebra-geometric approach (e.g., see \cite{BBEIM}),
the present formalism allows for the rigorous analytical treatment of infinite-genus potentials. 

We believe that the results of this work will lead to several applications in the future, including but not limited to
the efficient numerical calculation of finite-genus solutions with large genus
(using similar methods as in \cite{BilmanNabelekTrogdon} for the KdV equation), 
and the rigorous theory of soliton gases for the focusing NLS equation.

One immediate avenue for further study will be to establish a precise connection between the RHP in the present work for the finite-genus case and the RHP for the finite-genus solutions in (e.g., see \cite{KotlyarovShepelsky}).
We believe that a useful tool in this regard will be a transformation 
to move the poles of associated with the Dirichlet eigenvalues to the endpoints of the spectral bands,
similarly to what was done in \cite{BilmanNabelekTrogdon} for the KdV equation.

Another interesting avenue for future work will be a further extension of the formalism to remove some of the conditions in Assumption~\ref{assump}.
For example, recall that, in Assumption~\ref{assump}, we asked that all $\Gamma_n$ intersect the real axis. 
Next we briefly discuss the modifications that one would need to make in the formalism to remove that assumption.
There are three separate situations that need to be considered:

1. The first case that needs to be discussed is that of an arc $\Gamma_o$ located entirely in either the upper or lower half plane and which does not intersect either with the real axis nor any other $\Gamma_n$. 
In this case, the branch cuts for $\rho$ will be defined in the same way along the spectral bands, 
while the square root and fourth root branch cuts needed for $B(z)$ in \eqref{B} need to be modified. 
However, there is no guarantee that there will be the correct number of branch points along this $\Gamma_o$ to allow these square root and fourth root branch cuts to be entirely contained on $\Gamma_o$. 
To resolve this issue, one can run the branch cuts from any point on the leftmost spectral band on $\Gamma_o$ to the real axis and along the real axis to the point at infinity. 
A segment would need to be added to the the jump matrix $V(x,z)$ \eqref{e:Vexplicit} to account for these cuts, its behavior similar to when $k$ is even. Note that due to symmetry and the fact that bands not on the real axis must have a finite length, there must be an even number of fourth order branch points in the upper half plane and the same number in the lower half plane and the jump across the real axis due to the fourth order branch cuts will not be affected. 
The square root branch cuts brought about by Dirichlet eigenvalues on the band edges (which lack symmetry) can affect the jump across along the real axis, but that is easily handled by a counting function similar to $\mathcal{M}_n^{\pm}$ in~\eqref{mnplus} and $\mathcal{M}_n$ in~\eqref{mn}. 
The only required change is that $\mathcal{M}_n^{\pm}$ follows the cut from the real axis to the leftmost point of the left most spectral band and follows $\Gamma_o$ right to the rightmost of the right most spectral band.

2. The second case that needs to be discussed is that when an $\Gamma_{n_1}$ intersects another arc $\Gamma_{n_2}$ in such a way that no spectral bands intersect each other. 
In this case, the two intersecting arcs can be treated as two arcs that do not intersect any other $\Gamma_n$, 
and the modification described above can be used to define square roots, fourth roots and calculate the required jumps.

3. The third and final case to be discussed is that when an arc $\Gamma_{n_1}$ intersects another arc $\Gamma_{n_2}$ in such a way that spectral bands do intersect each other, resulting in a curved ``cross'' figure in the complex plane. 
For obvious reasons, this can only happen when one of these two arcs does not also intersect the real axis.
Let $\Gamma_{n_1}$ be this arc.
In this case, any other spectral bands along the two intersecting arcs can be treated as above, 
while the cross configuration is treated as a part of $\Gamma_{n_2}$. 
In this case, we define the terms ``below" and ``above" such that ``below'' denotes any locations along the branch cut for that 
$\Gamma_{n_2}$ between the real axis and the intersection point of the cross, while ``above'' denotes any locations on the other side of the intersection point of the cross. 
With the terms above and below defined as such, 
the jump of any quantity above the cross due to fourth order branch cuts is unchanged, 
while the jump of any quantity below the cross due to fourth order branch cuts is changed by a factor of $-1$.  
The jump on the arms of the cross due to fourth order branch cuts is $i$ in the lower half plane and $-i$ in the upper half plane. 
The modifications required for the counting functions $\mathcal{M}_n^{\pm}$ are conceptually simple. 
Every cross adds two band edges that a Dirichlet eigenvalue could reside on, which would result in square root branch cuts. These branch cuts would follow the arms of the cross to the intersected $\Gamma_n$ and move down to the real axis along which it would run to infinity. 
The jump along segment this branch cut runs would simply need to be changed by a factor of $-1$.

Finally, we also hope that this new formalism will make it possible to establish a more direct limiting connection between the problem on the line and that with periodic boundary conditions, 
as well as to use rigorous techniques of asymptotic analysis such as the Deift-Zhou method \cite{DeiftZhou1991} to 
compute semiclassical limits and advance our understanding of the behavior of solutions
with periodic boundary conditions.
We hope that the present work and the above discussion will spur further work on these and related topics.

\appendix
\section*{Appendix}
\setcounter{section}1
\setcounter{subsection}0
\setcounter{equation}0
\setcounter{figure}0
\def\thesection{\Alph{section}}
\def\theequation{\Alph{section}.\arabic{equation}}
\def\thetheorem{\Alph{section}.\arabic{theorem}}
\def\thefigure{\Alph{section}.\arabic{figure}}
\addcontentsline{toc}{section}{Appendix}

\subsection{Basic properties of the spectrum}
\label{a:spectrum}

The spectrum of the non-self-adjoint Dirac operator~\eqref{e:Diracoperator} is considerably more complicated than in the self-adjoint case.
Here we prove some basic properties that are used in the rest of the work.
The following proposition is straightforward:

\begin{proposition}
The fundamental matrix solution $Y(x,z)$ and modified fundamental matrix solution $\~Y(x,z)$ satisfy respectively the following symmetries:
\bse
\begin{gather}
Y^*(x,z^*)=\sigma_2 Y(x,z)\sigma_2,
\label{e:symY}
\\
-\sigma\~Y^*(x,z^*)\sigma= \~Y(x,z), 
\label{e:symY~}
\end{gather}
\ese
where 
\be
\sigma=\begin{pmatrix}0&1\\-1&0\end{pmatrix}\,.
\ee
\end{proposition}

In particular, \eqref{e:symY} follows from \cite{MaAblowitz}.
Since the fundamental matrix solution satisfies \eqref{e:symY}, 
using the definition~\eqref{e:Ytildedef} 
	of $\~Y$, we have
$	\sigma_2 Y^*(z^*)\sigma_2 = \sigma_2 U^{-1*}\~Y^*(z^*)U^*\sigma_2 = Y(z)=U^{-1}\~Y(z)U$,
which yields \eqref{e:symY~}.

\begin{proposition}
	For all $z\in\Complex$, the Floquet discriminant $\Delta(z)$ satisfies the Schwarz symmetry: $\Delta^*(z^*)=\Delta(z)$. As a result, for all $z\in\Complex/\Real$, the Floquet multiplier $\rho(z)$ satisfies $\rho^*(z^*)=\rho(z)$, the square root $r_{1/2}$ satisfies $r_{1/2}^*(z^*) = r_{1/2}(z)$ and $\z_k$ arise in complex conjugate pairs. 
\end{proposition}
\begin{proof}
	From \eqref{e:symY~}, we have that $\Delta^*(z^*)=\half \tr\~Y^*(z^*)=\half \tr(-\sigma\~Y(z)\sigma)=\half(\~y_{11}+\~y_{22})=\Delta(z)$, which directly gives the symmetry for $\rho(z)$. Additionally, since  $\Delta^2(\z_k)-1=\Delta^{*2}(\z_k^*)-1$, $\z_k$ appear in complex conjugate pairs. 
\end{proof}

\begin{proposition}
	If $\zeta_j$ is real, then it is at least a double root of $\Delta^2-1=0$. Moreover, for real $z$, $\Delta(z)$ is real and $-1\leq\Delta(z)\leq1$. In other words, there are no gaps on the real axis.
\end{proposition}
\begin{proof}
    Here we just need to prove that when $\zeta_j\in\Real$, $\Delta'(z)|_{z=\zeta_j}=0$.
    Denoting the columns of $Y(x,z)$ as $y_j(x,z)$ for $j=1,2$ as before and 
    expanding $y_1(x+L,z)$ in terms of $y_1(x,z)$ and $y_2(x,z)$, we have 
	\be
	y_1(x+L,z)=a y_1(x,z)+b y_2(x,z).
	\ee
	Evaluating $x$ at $0$, we obtain
	\be
	a(0,z)=y_{11}(L,z),\quad b(0,z)=y_{21}(L,z).
	\ee
	and from Wronskian relation, we have
	\be\label{Wr}
	a(0,z) a^*(0,z^*)+b(0,z) b^*(0,z^*)=1.
	\ee
	At $z=\zeta_j$, using the symmetry $y_{11}(x,z)=y_{22}^*(x,z^*)$, \eqref{Wr} becomes 
	\be
	y_{11,\re}^2(L,\zeta_j)+y_{11,\im}^2(L,\zeta_j)+y_{12,\re}^2(L,\zeta_j)+y_{12,\im}^2(L,\zeta_j)=1.
	\ee
	Moreover, at $z=\zeta_j,$ the Floquet discriminant $\Delta(\zeta_j)=\half\tr Y(L,\zeta_j)=y_{11,\re}(L,\zeta_j)$, and $\Delta^2(\zeta_j)-1=0$ yields $y_{11,\re}^2(L,\zeta_j)=1$. These conditions imply  $y_{11,\im}(L,\zeta_j)=y_{12,\re}(L,\zeta_j)=y_{12,\im}(L,\zeta_j)=0$. Consequently, the differentiation of \eqref{Wr} with respect to $z$ indicates that $\frac{\d\Delta}{\d z}|_{z=\zeta_j}=0$.

	Next, we prove that the entire real axis is part of the spectrum. Consider the symmetry of the monodromy matrix $M(z)$, 
	\be
    \label{symmetry}
	M(z)=\begin{pmatrix}
		M_{11}& M_{12}\\
		-M_{12}^* & M_{11}^*
		\end{pmatrix}\,,\qquad z\in\Real\,.
	\ee
	For real $z$, we always have that $\Delta(z)$ is real. Moreover, we have $1=\det M=|M_{11}|^2+|M_{12}|^2$, which implies $|M_{11}|\leq1$. Therefore, $4|\Delta|^2=|M_{11}+M_{11}^*|^2\leq4|M_{11}|^2\leq 4$. This indicates that there are no gaps on the real axis.
\end{proof}

\begin{proposition}
Under Assumptions~\ref{assump}, there are no real branch points for $(\Delta^2-1)^{1/2}$. 
\end{proposition}

\begin{proposition}\label{RealZeta}
	Each real main eigenvalue $\zeta_j$ is also a Dirichlet eigenvalue. 
\end{proposition}
\begin{proof}
	For a real eigenvalue $\zeta$, $|\Delta(\z)|$ can be written as
	\be\label{e:A.51}
	|\Delta(\z)|=\half\left|\~M_{11}(\z)+\~M_{22}(
	\z)\right|=\half\left|\~M_{11}(\z)+\~M_{11}^*(
	\z)\right|=|\Re \~M_{11}(\z)|=1\,,
	\ee
    indicating that $|\~M_{11}(\z)|\geq 1$.
	Moreover, since $\~M(z)$ has a unit determinant, we have
	\be\label{e:A.52}
	1=\det\~M=\~M_{11}\~M_{11}^*+\~M_{12}\~M_{12}^*=|\~M_{11}|^2+|\~M_{12}|^2,
	\ee
    which implies $|\~M_{11}(\z)|\leq 1$.
    Therefore, $|\~M_{11}(\z)|=1$, and $|\~M_{12}(\z)|=0$. implying that $\z$ is a Dirichlet eigenvalue.
\end{proof}

Proposition~\ref{RealZeta} shows that all real main eigenvalues have Dirichlet eigenvalues locked to them. 
Note however that the converse of Proposition~\ref{RealZeta} does not hold.  Namely, 
there can exist real Dirichlet eigenvalue that do not coincide with main eigenvalues.

\subsection{Alternative sets of Dirichlet eigenvalues and trace formulae}
\label{a:alternative}

Similar to our previous work \cite{BiondiniZhang2023}, only one set of Dirichlet eigenvalues suffices to reconstruct the potential  $q(x)$.
We now show that an alternative set of Dirichlet eigenvalues can also be used for inverse spectral problem.

As introduced in~\cite{MaAblowitz,McLaughlinOverman}, this second set is defined analogously to \eqref{e:dirichlet} but with modified boundary conditions:
\be
v_1(x_0)+\i v_2(x_0)=v_1(x_0+L)+\i v_2(x_0+L)=0.
\label{e:DirichletBC2}
\ee
Let $\{\mu_j(x_o)\}_{j\in\mathbb{Z}}$ and $\{\check\mu_j(x_o)\}_{j\in\mathbb{Z}}$ 
denote the Dirichlet eigenvalues corresponding to \eqref{e:Dirbcs} and \eqref{e:DirichletBC2}, respectively. 
Then (e.g., see \cite{MaAblowitz,McLaughlinOverman})
\bse
\label{e:trace}
\begin{gather}
q(x) + q^*(x) = \i \sum_{j\in\mathbb{Z}} (\z_{2j} + \z_{2j+1} - 2\mu_j(x))\,,
\label{e:trace1}
\\
q(x) - q^*(x) =  \sum_{j\in\mathbb{Z}} (\z_{2j} + \z_{2j+1} - 2\check\mu_j(x))\,.
\label{e:trace2}
\end{gather}
\ese

To use the second set of Dirichlet eigenvalues to construct a Riemann-Hilbert problem, define a modified similarity transformation:
\be
\check{Y}=\check{U}Y\check{U}^{-1}\,,\qquad
\check{U}=U\e^{-\frac{\i\pi}{4}\sigma_3}\,.
\ee
The modified solutions $\check{Y}$ satisfy
$\check{Y}_x = \check{U}(-\i z\sigma_3 + Q)\check{U}^{-1}\,\check{Y}$,
and the Floquet discriminant and multipliers remain unchanged: $\check{\Delta}=\Delta$, and $\check{\rho}=\rho$. Hence, the main spectrum $\{\zeta_j\}$ preserved.
On the other hand, the alternative Dirichlet spectrum consists of zeros of $\check{y}_{12}(L,\check{\mu})$,
and satisfies all the properties in Theorem~\ref{spectraprop}. We define the auxiliary spectral data:
\be
\check{S}(q):=\{{E}_{2k-1},{E}_{2k},\check{\mucirc}_k,\check{\nu}_k\}_{k=\g_-,\dots,\g_+}
\ee
with $\check{\mucirc}_k$ and $\check{\nu}_k$ defined analogously to Definitions \ref{def:dirichleteigenvalue} and \ref{nudef}. 
By standard Bloch-Floquet theory, one obtains modified Bloch-Floquet solutions $\check{\psi}^\pm$ using $\check{y}_{ij}$, leading to a new RHP 
$\check{\Phi}$ with asymptotics
\be
\check\Phi(x,z) = \check{U}\,(I+O(1/z))\,\check{B}(z)\,,
\ee
where $\check{B}(z)$ is defined in \eqref{B} with $\mu_n$ replaced by $\check{\mu}_n$. Following the same procedure as in Sections~\ref{s:asymptotics}--\ref{s:time}, this RHP yields a unique solution, 
from which the potential can be reconstructed.

In closing this section, it might be helpful to clarify that, in many other works in the literature, yet a different definition is used
e.g., see~\cite{gesztesyweikard_acta1998,gesztesyweikard_bams1998,AdvMath2023},
and the Dirichlet eigenvalues are defined as those values $z\in\C$ for which the solutions of the Zakharov-Shabat problem satisfy the boundary conditions
\vspace*{-1ex}
\be
v_1(x_o) = v_1(x_o + L) = 0\,.
\label{e:GesztesyDirEvals}
\ee
The reason why this is relevant is that, 
for finite-genus solutions, 
the number of movable Dirichlet eigenvalues defined according to the BCs~\eqref{e:GesztesyDirEvals}
is exactly equal to the genus,
as shown in \cite{gesztesyweikard_acta1998,gesztesyweikard_bams1998},
In contrast, the number of movable Dirichlet eigenvalues defined according to 
either the BCs~\eqref{e:Dirbcs} equals twice the genus, 
as is the number for those defined according to~\eqref{e:DirichletBC2}.

\subsection{Explicit formulation of the RHP for genus-zero potentials}
\label{s:genuszero}

Let $q(x) = A\,\e^{\i\alpha}$, with $\alpha\in\Real$ and where we can take $A>0$ without loss of generality.
Since $q(x)$ is independent of $x$, the scattering problem can be solved exactly.
Let us choose the eigenvector matrix of $X = - \i z\sigma_3 + Q$ as
\be
W(z) = I - \i\sigma_3Q/(z+\lambda),
\ee
so that $X\,W = W\,(-\i\lambda)\sigma_3$, 
where $\lambda = \lambda(z) = (z^2+A^2)^{1/2}$.
Explicitly, we define the complex square root so that:
(i) its branch cut is $[-\i A, \i A]\cup\R$;
(ii) $(z^2+A^2)^{1/2}>0$ for all $z\in\i\Real\setminus[-\i A,\i A]$;
(iii) on $[-\i A, \i A]$, $(z^2+A^2)^{1/2}$ is continuous from the left, and on the real axis, it is continuous from above.

For illustration purposes, 
we will slightly generalize some of the quantities defined in section~\ref{s:spectrum}. 
We do so even though it is not necessary for the formulation of the RHP, 
in order to identify the movable and immovable Dirichlet eigenvalues. 
Accordingly, let $Y(x,x_o,z)$ be the fundamental matrix solution of the ZS problem such that $Y(x_o,x_o,z) = I$ at 
an arbitrary base point $x_o$.
In light of the above calculations, it is straigthforward to see that
\bse
\begin{gather}
	Y(x,x_o,z) = W(z)\,\e^{-\i\lambda (x-x_o) \sigma_3}W^{-1}(z)\,,
	\label{e:Ygenus0}
	\\
\intertext{and the corresponding monodromy matrix is}
    M(x_o,z) = Y(x_o+L,z) = \frac1{(\l+z)^2+A^2} 
	\begin{pmatrix}\e^{-\i\l L}\big((\l+z)^2+A^2\e^{2\i\l L}\big) & 2A\e^{\i\alpha}(\l+z)\sin(\l L) \\ 
		-2A\e^{-\i\alpha}(\l+z)\sin(\l L) & \e^{-\i\l L}\big(A^2+\e^{2\i\l L}(\l+z)^2\big)
	\end{pmatrix},
\end{gather}
\ese
implying $\D(z) = \cos(\lambda L)$.
The periodic eigenvalues are therefore the points $z_n$ such that $\lambda(z_n) = 2n\pi/L$, $n\in\Integer$.
All such eigenvalues are double zeros of $\D(z)$. 
Note that, if $AL>\pi$, some of these double points lie on the imaginary axis, which violates clause~(ii) of Assumption~\ref{assump}.
Nonetheless, we will see below that this is not an obstacle to the formulation of the RHP.
The existence of such complex double points is related to the modulational instability of the constant potential,
as shown for example in \cite{ForestLee}.
It is also straightforward to see that, if $AL = n\pi$, the point $z=0$ is a main eigenvalue with multiplicity~4.

Next we proceed to the calculation of the Dirichlet eigenvalues. 
By analogy with~\eqref{e:Ygenus0}, we generalize \eqref{e:Ytildedef} by defining
\be
\~Y(x,x_o,z) = U Y(x,x_o,z) U^{-1}\,,
\label{e:YildexOdef}
\ee
It is then straightforward to see that the Dirichlet eigenvalues with base point $x_o$ are the zeros of $\~y_{12}(x_o+L,x_o,z)$.
Explicitly, we have
\bse
\label{ytildeconstant}
\begin{gather}
\~Y(x_o+L,x_o,z) = \begin{pmatrix} 
    \displaystyle \cos(\l L)+\frac{\sin(\l L)}{\l}\i A\sin\alpha &
    \displaystyle \frac{\sin(\lambda L)}\lambda\,(z-\i A\cos\alpha)
	\\
    \displaystyle	- \frac{\sin(\lambda L)}\lambda\,(z+ \i A\cos\alpha) &
    \displaystyle	\cos(\l L) - \frac{\sin(\l L)}\l \i A\,\sin\alpha\
    \end{pmatrix}
\end{gather}
\ese
We then see that there are Dirichlet eigenvalues wherever $\lambda L = n\pi$ for $n\in\Integer\setminus\{0\}$,
whereas the points $z = \pm i A$, at which $\lambda = 0$, are not zeros of $\~y_{12}(x_o+L,x_o,z)$
unless $\alpha=0$ or $\alpha = \pi$.
In addition, the point $z=0$ is a multiple Dirichlet eigenvalue if $AL= n\pi$.
All these Dirichlet eigenvalues are single and tied to the corresponding double main eigenvalues discussed earlier.
There is also an additional, single Dirichlet eigenvalue  $\mu_o = \i A\,\cos\alpha$, which, as the above calculations show, is also independent of $x_o$ and therefore immovable. 
The numerically computed Lax spectrum and Dirichlet eigenvalue are shown in Figure~\ref{f:genuszero}
(see Appendix~\ref{s:numerics} for details).
In any case, as discussed in section~\ref{s:inverse}, for the present approach to the IST it is sufficient to use 
the Dirichlet eigenvalues with base point~$x_o=0$.

\begin{figure}[b!]
\centerline{\includegraphics[trim= 0 0 0 0,clip,width=\figwidth]{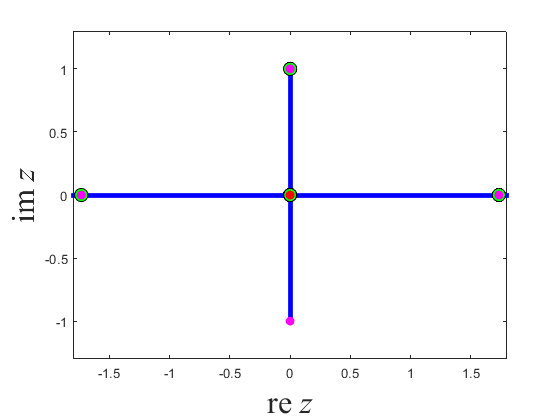}\kern-1em
  \includegraphics[trim= 0 0 0 0,clip,width=\figwidth]{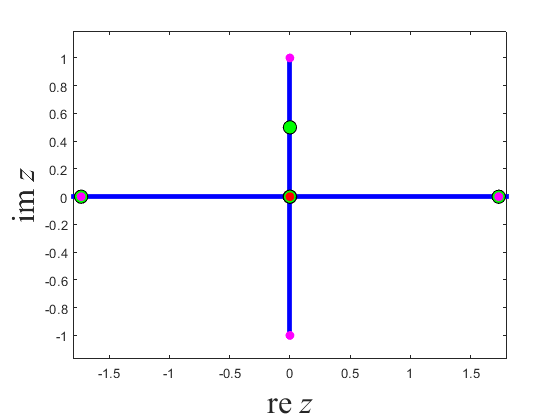}}
\caption{Left: The Lax spectrum (in blue) for $q(x)=\dn(x,0)=1$. 
The periodic and antiperiodic eigenvalues are highlighted in magenta and red, respectively, and the Dirichlet eigenvalues are shown as large green circles outlined in black. 
Here the period $L$ was chosen to correspond to the period of $\dn(x,0)$, namely $L=\pi$, 
and the 
base point was $x_0=0$. 
Right: Same, but for $q(x)=\e^{\i\pi/3}$. 
Note how the Dirichlet eigenvalue is located at $z= \i/2$, 
in agreement with \eqref{ytildeconstant}.}
\label{f:genuszero}
\end{figure}

Equations~\eqref{BF} and~\eqref{e:YildexOdef} with $x_o=0$ immediately yield explicit expressions for the Bloch-Floquet solutions explicitly.
However, the resulting formulae are complicated and not particularly illuminating, and are therefore omitted for brevity.
Since $\rho(\mu_o)=\cos(\l(\mu_o) L)-\i\sin(\l(\mu_o) L)$, which coincides with $\~y_{22}(L,\mu_o)$, by definition we have $\nu_0 = 1$ when $\alpha\in(0,\pi/2)\cup(3\pi/2,2\pi)$ and $\nu_0 = -1$ when $\alpha\in(\pi/2,3\pi/2)$. 
Let us consider without loss of generality the case $\nu_0=1$ and compute the jump matrix of $\Phi(x,z)$ as defined in \eqref{Phi}. 
When $\nu_0=1$, we have $f^-=1$, $f^+=\mu_o-z$, and $f^0=\~y_{12}(L,z)/(f^-f^+)=-\sin(\l L)/\l$. 
Thus, we obtain
\bse
\begin{gather}
B(z) = \begin{cases}
	\displaystyle \i b(z) \diag(1,\mu_o-z), &z\in\Complex^+,
	\\[0ex]
	\displaystyle b(z) \diag(\mu_o-z,1), &z\in\Complex^-\,,
\end{cases}
\label{e:Bgenus0}
\end{gather}
\ese
with $b(z) = \e^{\i\pi/4}/(z^2+A^2)^{1/4}$.
For all $z\in \Sigma(\L)$, we have $\rho_+(z)=\rho_-^{-1}(z)$ and $\l_+(z)=\l_-(z)$,
where as usual the subscripts $\pm$ denote the non-tangential limits.
As $z\in\R$, from \eqref{e:defPsi} and \eqref{BF}, there is a switch between the first and second columns of $\Psi$ in different half-planes, which cancels the jump of $\rho$, implying that $\Psi_+=\Psi_-$. Thus, we obtain the jump as
\be\label{VR}
\Phi^{-1}_-(x,z)\Phi_+(x,z) = \e^{-\i zx\hat{\sigma}_3}(B_-^{-1}B_+)
  =  \i  \frac{ b_+(z) }{ b_-(z) } 
      \begin{pmatrix} \frac{1}{\mu_o-z} & 0\\ 0 & \mu_o-z \end{pmatrix}\,,
      \qquad z\in\Real\,.
\ee
For $z\in[-\i A, \i A]$, the definition of $\Psi$ is the same on both sides of the contour in each half-plane. 
However, the jumps of $\rho$ and $\l$, result in a switch of the columns of $\Psi(x,z)$, which yields $\Psi_+(x,z)=\Psi_-(x,z)\sigma_1$. The jump is therefore given by
\vspace*{-1ex}
\begin{align}
\label{ViR}
\Phi^{-1}_-(x,z)\Phi_+(x,z) = \begin{cases} \displaystyle
\frac{ b_+(z) }{ b_-(z) } \,\e^{-\i zx\hat\sigma_3}
	\begin{pmatrix}	0 & \mu_o-z\\	\frac{1}{\mu_o-z} & 0 \end{pmatrix}, 
        \quad z\in(0,\i A),\\ \displaystyle
\frac{ b_+(z) }{ b_-(z) } \,\e^{-\i zx\hat\sigma_3}
        \begin{pmatrix} 0 & \frac{1}{\mu_o-z}\\   	\mu_o-z & 0    \end{pmatrix}, 
        \quad z\in(-\i A,0).
    \end{cases}
\end{align}
To make these jumps fully explicit, it remains to consider the quartic root. 
From the above definition of $\lambda(z)$, we have
\vspace*{-1ex}
\be\label{4throotjump}
\frac{ b_+(z) }{ b_-(z) } = \begin{cases}
    -\i, & z\in(-\infty,0)\cup(-\i A,0)\,,\\
    \i, & z\in(0,\infty)\cup(0,\i A)\,.
\end{cases}
\ee
Inserting~\eqref{4throotjump} into \eqref{VR} and \eqref{ViR}, we get \eqref{e:Vg0}, which also coincides with the definition in \eqref{e:Vexplicit}.
Hence, the $2\times2$ matrix-valued function $\Phi(x,z)$ defined by~\eqref{Phi} indeed solves the following RHP:
\begin{RHP}
\label{RHPgenus0}
Find a $2\times2$ matrix-valued function $\Phi(x,z)$ such that
\vspace*{-1ex}
\begin{enumerate}
\advance\itemsep-4pt
    \item 
	$\Phi(x,z)$ is a holomorphic function of $z$ for $z\in\C\setminus(\R\cup[-\i A,\i A])$. 
    \item 
        The non-tangential limits $\Phi_\pm(x,z)$ to $\R\cup[-\i A,\i A]$ are continuous functions of $z$ in $\R\cup(-\i A,\i A)$, and have at worst quartic root singularities at $\{\pm \i A\}$. 
    \item 
	$\Phi_\pm(x,z)$ satisfy the jump relation 
		\be
			\Phi_+(x,z)=\Phi_-(x,z)V(x,z),\qquad z\in \R\cup[-\i A,\i A] \,,
		\ee
         with 
    \be
    \label{e:Vg0}
    V(z)=\begin{cases}
	\displaystyle\begin{pmatrix} 0 & -\i\e^{-2\i zx}(\mu_o-z) \\ -\i\e^{2\i zx}/(\mu_o-z) & 0 \end{pmatrix},&z\in(0,\i A)\,,
	\\[2ex]
	\begin{pmatrix} 0 & \i\e^{-2\i zx} /(\mu_o-z) \\ \i\e^{2\i zx} (\mu_o-z) & 0 \end{pmatrix},&z\in(-\i A,0)\,,
        \\[2ex]
        \begin{pmatrix}
            \frac{1}{z-\mu_o} & 0\\
            0 & z-\mu_o
        \end{pmatrix}\,, & z\in (-\infty,0)\,,
        \\[2ex]
        \begin{pmatrix}
            \frac{1}{\mu_o-z} & 0\\
            0 & \mu_o-z
        \end{pmatrix}\,, & z\in (0,\infty)\,.
    \end{cases}
    \ee
    
    \item 
	As $z\to\infty$ with $z\in\Complex$, $\Phi(x,z)$ has the following asymptotic behavior 
		\be
			\Phi(x,z) = (I+O(1/z))B(z), 
			\label{asymPhi_g0}
		\ee
	with $U$ as in~\eqref{e:Udef} and $B(z)$ as in~\eqref{e:Bgenus0}.
    \item 
	There exist positive constants $c$ and $M$ such that $|\phi_{ij}(x,z)|\leq M\e^{c|z|^2}$ for all $z\in \mathcal{D}$.
	\end{enumerate}
\end{RHP}

In the special case $\alpha=2n\pi$, the potential is real and even, and the above formalism simplifies considerably.
Indeed, in this case, the Dirichlet eigenvalue $\mu_0$ equals $iA$ (consistently with Lemma~\ref{l:Dirichletbase0}).
Thus, $f^\pm(z)\equiv1$, and $b(z) = \e^{i\pi/4} [(z-\i A)/(z+\i A)]^{1/4}$.
The jump matrix in the RHP~\label{RHPgenus0} then simplifies to
\be
V(z) = \begin{pmatrix} 0 & \i\,\e^{-2\i zx} \\ \i\,\e^{2\i zx}& 0 \end{pmatrix}, \qquad z\in(-\i A,\i A)\,,
\ee
while $V(z) \equiv I$ for $z\in\Real$\,.

\subsection{Explicit formulation of the RHP for a genus-one elliptic potential}
\label{s:dn}

The focusing NLS equation admits quite a large family of genus-one elliptic solutions, some of which were studied in 
\cite{MaAblowitz}, while the full family was later characterized in \cite{Kamchatnov,DeconinckSegal}.
One simple such solution is the potential $q(x) = \dn(x,m)$, where dn is one of the Jacobi elliptic functions \cite{NIST}
and $m\in[0,1]$ is the elliptic parameter.
The more general family $q(x) = A\,\dn(x,m)$ with $A>0$ was studied in \cite{AdvMath2023}, where it was shown that
if and only if $A\in\Natural$, it is a finite-genus potential with genus equal to $2A-1$.
Moreover, in \cite{AdvMath2023} it was shown that, when $A\in\Natural$, for all $m\in(0,1)$ the Lax spectrum comprises the real $z$ axis plus $2A$ bands on $i\Real$ (symmetrically located with respect to the real axis as dictated by Schwarz symmetry), separated by $2A-1$ spectral gaps.
Here we consider the case $A=1$ and study the Lax spectrum, eigenfunctions, Dirichlet eigenvalues and the corresponding RHP formulation.
(It was also shown in \cite{AdvMath2023} that each of the $2A-1$ spectral gaps on $i\R$ contains exactly one movable Dirichlet eigenvalue defined according to the BCs~\eqref{e:GesztesyDirEvals}, but that is not relevant for the formalism presented in this paper.)

Note first that the main eigenvalues can be obtained using either the methods of~\cite{BertolaTovbis2016} or those of~\cite{DeconinckSegal}.
With either method, one finds that the main eigenvalues in $\C^+$ are located at 
$\z_{1,2}=\frac \i2\big(1\pm\sqrt{1-m}\big)$,
with symmetrically located eigenvalues at $-\z_{1,2}$ in $\C^-$.
On the other hand, the determination of the Dirichlet eigenvalues is more involved, and, as far as we know, has not been presented in the literature before.

By Lemma~\ref{l:Dirichletbase0}, we know that the Dirichlet eigenvalues with base point $x_o=0$ are located at some points of the main spectrum.
The trace formulae~\eqref{e:trace} then allow us to determine exactly which of the main eigenvalues they are.  In particular, they coincide with $\z_{1,2}$.
Next, even though it is not necessary for the formulation of the RHP, for illustrative purposes we present the explicit calculation of the ZS eigenfunctions and we obtain an equation that yields the location of the Dirichlet eigenvalues with arbitrary base point~$x_o$,
similarly to the previous section.
The following calculations are based on the framework of \cite{DeconinckSegal}.

Note first that the solution of the NLS equation with IC $q(x,0) = \dn(x,m)$ is simply
$q(x,t) = \e^{2\i\omega t}\,\dn(x,m)$, with $\omega = 1 - m/2
$.
Accordingly, one can look for simultaneous solutions $\Xi(x,t,z)$ of both parts of the modified Lax pair
\be
\Xi_x = \~X(x,z)\,\Xi\,,\qquad
\Xi_t = [\~T(x,z) - \i\omega\sigma_3]\,\Xi\,,
\ee
with $\~X(x,z)$ and $\~T(x,z)$ given by~\eqref{e:NLSLP} but with $q(x,t)$ replaced by $q(x) = \dn(x,m)$,
which makes the resulting expressions time-independent.
Further, note that $\~T(x,z) + \i\omega\sigma_3$ can be written as  
\bse
\begin{gather}
\~T(x,z) - \i\omega\sigma_3 = \begin{pmatrix} A & B \\ C & -A \end{pmatrix},
\\
A(x,z) = - 2\i z^2 + \i q^2(x) - \i \omega\,,\qquad
B(x,z) = 2z q(x) + \i q'(x)\,,\qquad C(x,z) = B(x,-z)\,.
\end{gather}
\ese
Since $A$, $B$ and $C$ are time-independent, we can therefore separate variables by letting
$\Xi(x,t,z) = \e^{\Omega t}\chi(x,z)$.
Then $\chi(x,z)$ solves the homogeneous system of algebraic equations 
\vspace*{-1ex}
\be
(\~T(x,z) - \i\omega - \Omega I)\,\chi(x,z) = 0\,.
\label{e:g1linearsystem}
\ee
In order for nontrivial solutions to exist, one obviously needs $\det( \~T(x,z) + \i\omega - \Omega I ) = 0$, i.e., 
$\Omega^2 = A^2(x) + B(x)\,C(x) = - 4z^4 - \txtfrac14( m - 4 z^2)^2$, 
which yields two values for $\Omega$:  
$\Omega_\pm  = \pm \frac12 \big[ - 16z^2 - ( m - 4 z^2)^2 \big]^{1/2}\,$.
Note that $\Omega$ is not only independent of~$t$, but also of~$x$.
In other words, $\Omega$ is strictly a function of $z$ and $m$.
For each choice of $\Omega$, we can then solve~\eqref{e:g1linearsystem} to obtain
\bse
\begin{gather}
\chi_\pm (x,x_o,z) = \e^{g_{\pm,x_o}(x,z)} \begin{pmatrix} B(x,z) \\ \Omega(z) - A(x,z) \end{pmatrix}\,,
\intertext{with}
g_{\pm,x_o}(x,z) = \int_{x_o}^x \frac{(\Omega_\pm(z)-A(s,z))\,q(x) - B_x(s,z) - \i z }{B(s,z)} \,\d s\,.
\end{gather}
\ese
We can then construct a fundamental matrix solution of the ZS problem as
\bse
\begin{gather}
W(x,z) = \big( \chi_-(x,z), \chi_+(x,z) \big)\,, 
\\
\intertext{which in turn allows us to obtain $Y(x,x_o,z)$ as}
Y(x,x_o,z) = W(x,x_o,z)\,W^{-1}(x_o,x_o,z)\,.
\end{gather}
\ese
Straightforward calculations then yield the Floquet discriminant as
\be
\Delta(z,m) = \frac12\big( \e^{g_{-,0}(2K_m,z)} + \e^{g_{+,0}(2K_m,z)} \big)\,.
\ee
where $K_m = K(m)$ is the complete elliptic integral of the first kind \cite{NIST}.
Finally, defining $\~Y(x,x_o,s)$ as in \eqref{e:YildexOdef}, one explicitly obtains
\bse
\begin{gather}
\~y_{12}(x,x_o,z) = \i \left(\e^{g_{-,0}(2K_m)}-\e^{g_{+,0}(2K_m)}\right) 
-\frac{C_-(x_o,z)C_+(x_o,z)}{2(\Omega_--\Omega _+) (2 z\dn(x_o,m)-\i m \cn(x_o,m) \sn(x_o,m))}\,,
\label{e:dnDirichleteigs}
\\
\intertext{where}
C_\pm(s,z) = 2 z^2 -2 \i z \dn(s,m) - m \cn(s|m) \sn(s,m)-\dn(s,m)^2+\omega -\i \Omega_\pm\,.
\end{gather}
\ese
The zeros of $\~y_{12}(x_o+L,x_o,z)$ then determine the Dirichlet eigenvalues with base point~$x_o$.
Although an analytical expression is not available,
it is straightforward to obtain the zeros of~\eqref{e:dnDirichleteigs} numerically.
The results agree very well with the numerically obtained Lax spectrum and Dirichlet eigenvalues computed according to the numerical methods described in section~\ref{s:numerics}
(see Figs.~\ref{f:genusone} and~\ref{f:Dirichletmotion}).

Finally, we formulate the RHP satisfied by $\Phi(x,t,z)$.
Let us denote the two movable Dirichlet eigenvalues by $\mu_0$ and $\mu_1$, and consider the case where $\mu_0=\z_1$, and $\mu_1=\z_2$, which happens when $x_o=0$ (corresponding  to the top left plot in Figure~\ref{f:genusone}). Note that, in this case, both $f^-$ and $f^+$ are simply equal to $1$, and $f^0=\~y_{12}(L,z)/2$. 
Then, $B(z)$ is given explicitly by
\bse
\begin{gather}
B(z) = \begin{cases}
	\displaystyle \i b(z), &z\in\Complex^+,
	\\[0ex]
	b(z), &z\in\Complex^-\,,
\end{cases}
\label{e:Bgenus1}
\end{gather}
\ese
where now
\be
b(z) = \e^{\i\pi/4}\bigg[\frac{(z-\zeta_1)(z-\zeta_2)}{(z+\zeta_1)(z+\zeta_2)}\bigg]^{1/4}\,.
\label{e:bg1}
\ee
Note that $B(z)$ is a scalar since the potential is real and even, in agreement with Proposition~\ref{l:Brealevenq}.
We define the complex square root in the numerator of~\eqref{e:bg1} so that:
(i) its branch cut is $[-\z_2, -\z_1]\cup[\z_1,\z_2]\cup\R$;
(ii) $[(z-\mu_0)(z-\mu_1)]^{1/2}>0$ for all $z\in\i\Real\setminus\left([-\z_2, -\z_1]\cup[\z_1,\z_2]\right)$;
(iii) on $[-\z_2, -\z_1]\cup[\z_1,\z_2]$, $[(z-\mu_0)(z-\mu_1)]^{1/2}$ is continuous from the left, and on the real axis, it is continuous from above. 

Now we compute the jump matrix of $\Phi(x,z)$ defined in \eqref{Phi}. 
When $z\in \Sigma(\L):=[-\z_2, -\z_1]\cup[\z_1,\z_2]\cup\R$, we have $\rho_+(z)=\rho_-^{-1}(z)$.  
For $z\in\R$, we have $\Psi_+(x,z) = \Psi_-(x,z)$,
by the same arguments as in Appendix~\ref{s:genuszero}.
Thus, the jump is, again,
\be\label{VRg1}
V(x,z)=\e^{-\i zx\hat{\sigma}_3}(B_-^{-1}B_+)= \i \frac{b_+(z)}{b_-(z)}\,,\qquad z\in\R\,.
\ee
For $z\in (-\z_2,-\z_1)\cup(\z_1,\z_2)$ along $\i\R$, just as in the genus-zero case, we have $\Psi_+(x,z)=\Psi_-(x,z)\sigma_1$. 
The jump is therefore given by
\vspace*{-1ex}
\be
\label{ViR1g1}
V(x,z)=
\frac{b_+(z)}{b_-(z)} \e^{-\i zx\hat\sigma_3}\sigma_1\,,\qquad
z\in (-\z_2,-\z_1)\cup(\z_1,\z_2)\,.
\ee
Finally, for $z\in(-\z_1,\z_1)$ along $\i\R$, $\Psi(x,z)$ has no jump, since there is no jump for $\rho$. 
Therefore, only the branch cuts of the quartic root contribute to the jump, implying:
\be\label{ViR2g1}
V(x,z) = \frac{b_+(z)}{b_-(z)}\,,
\qquad 
z\in(-\z_1,\z_1).
\ee

It thus remains to compute the jump of $b(z)$ in order to define the RHP explicitly.
Straightforward calculations show that:
\vspace*{-1ex}
\be\label{4throotjumpg1}
\frac{ b_-(z) } { b_+(z) } 
  = \begin{cases}
    \i, & z\in\R\cup(-\z_2,-\z_1)\cup(\z_1,\z_2)\,,\\
    -1,& z\in (-\z_1,\z_1)\,.
\end{cases}
\ee
Inserting~\eqref{4throotjumpg1} into \eqref{VRg1}, \eqref{ViR1g1} and \eqref{ViR2g1}, we get \eqref{e:Vg1}, which also coincides with the definition in \eqref{e:Vexplicit}.
To summarize, we arrive at the
 $2\times2$ matrix-valued function $\Phi(x,z)$, which solves the following RHP:
\begin{RHP}
\label{RHPgenus0}
Find a $2\times2$ matrix-valued function $\Phi(x,z)$ such that
\vspace*{-1ex}
\begin{enumerate}
\advance\itemsep-4pt
    \item 
    $\Phi(x,z)$ is a holomorphic function  of $z$ for $z\in\C\setminus(\R\cup[-\zeta_2,\zeta_2])$. 
    \item 
    The non-tangential limits $\Phi_\pm(x,z)$ of $\Phi(x,z)$ to $\R\cup[-\zeta_2,\zeta_2]$ are continuous functions of $z$ in  $(\R\cup[-\zeta_2,\zeta_2])\setminus\{\pm\zeta_1,\pm\zeta_2\}$, and have at worst quartic root singularities 
    at $z = \pm \zeta_1,\pm\zeta_2$. 
    \item 
    $\Phi_\pm(x,z)$ satisfy the jump relation 
		\be
			\Phi_+(x,z)=\Phi_-(x,z)V(x,z),\qquad z\in \R\cup[-\zeta_2,\zeta_2]\,,
		\ee
         with 
    \be
    \label{e:Vg1}
    V(z)=\begin{cases}
	\displaystyle\begin{pmatrix} 0 & \i\,\e^{-2\i zx} \\ \i\,\e^{2\i zx} & 0 \end{pmatrix},&z\in(\z_1,\z_2)\,,
	\\[2ex]
	\begin{pmatrix} 0 & \i\,\e^{-2\i zx}\\ \i\,\e^{2\i zx} & 0 \end{pmatrix},&z\in(-\z_2,-\z_1)\,,
        \\[2ex]
        -I\,, & z\in (-\z_1,\z_1)\,.
    \end{cases}
    \ee
    \item 
    As $z\to\infty$ with $z\in\Complex$, $\Phi(x,z)$ has the following asymptotic behavior 
		\be
			\Phi(x,z) = (I+O(1/z))B(z), 
			\label{asymPhi_g1}
		\ee
	with $B(z)$ as in~\eqref{e:Bgenus1}.
    \item 
    There exist positive constants $c$ and $M$ such that $|\phi_{ij}(x,z)|\leq M\e^{c|z|^2}$ for all $z\in \mathcal{D}$.
\end{enumerate}
\end{RHP}

\begin{figure}[t!]
\centerline{\includegraphics[trim= 0 0 0 0,clip,width=\figwidth]{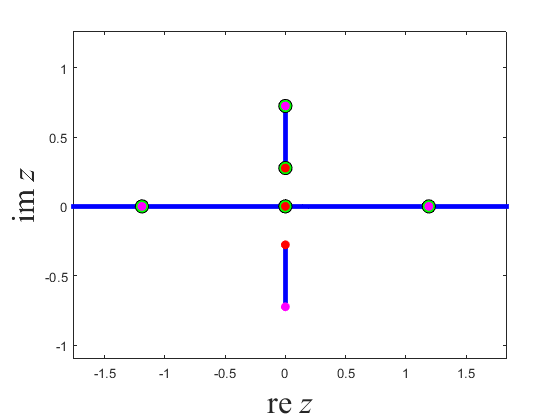}\kern-1em
  \includegraphics[trim= 0 0 0 0,clip,width=\figwidth]{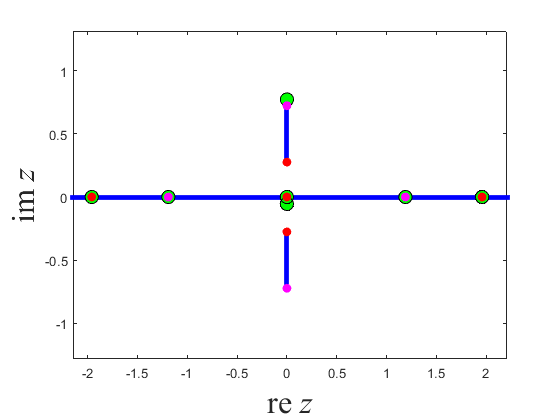}}
\centerline{\includegraphics[trim= 0 0 0 0,clip,width=\figwidth]{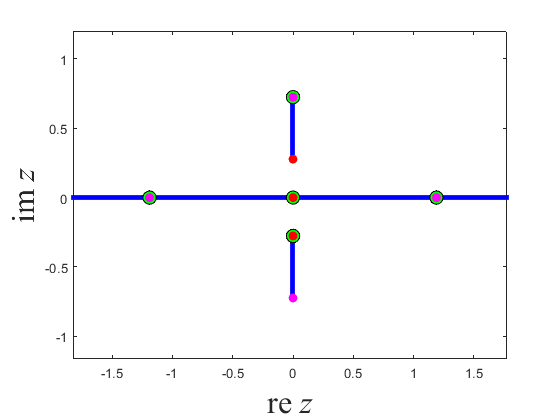}\kern-1em
  \includegraphics[trim= 0 0 0 0,clip,width=\figwidth]{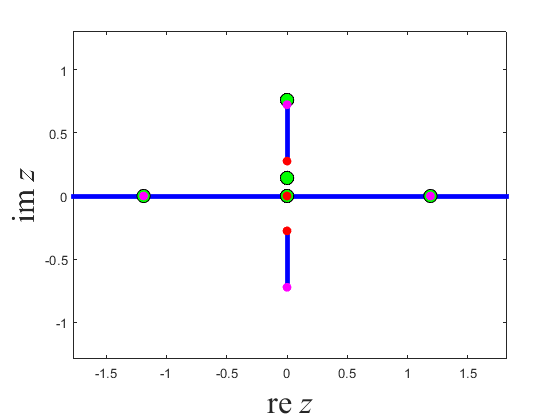}}
\caption{Plots of $q(x)=\dn(x,m)$ for $m = 0.8$ and  various choice of base points. 
Top left: $x_0=0$. 
Top right: $x_0=1$. 
Bottom left: $x_0=L/2$.
Bottom right: $x_0=4$. 
Here $L=2K(m)$, where $K(m)$ is the complete elliptic integral of the first kind.}
\label{f:genusone}
\end{figure}

In closing this section, we should note that the focusing NLS equation admits a larger family of genus-one solutions \cite{DeconinckSegal}.
Even though in this appendix we limited ourselves to treating the dn potential for brevity,
we believe that the above approach will work for all of them.

\subsection{A numerical approach for the calculation of the Dirichlet spectrum}
\label{s:numerics}

The main spectrum of the scattering problem can be calculated using the Floquet-Hill method \cite{DK2006}, similarly to what was done in~\cite{BOT_JST2023}. 
(Specifically, periodic and anti-periodic eigenvalues can be computed by taking the Floquet exponent $\nu$ to be $0$ and $\pi/L$ respectively.)
Here we discuss how one can efficiently compute the Dirichlet eigenvalues.

The calculation of the Dirichlet spectrum begins with the modified scattering problem~\eqref{e:modZS}. 
Recall that the Dirichlet eigenvalues with base point $x_0$ are the zeros of $\~y_{12}(x_0+L,z)$,
where $\~Y(x,z) = U Y(x,x_0,z) U^{-1}$ and 
$Y(x,x_0,z)$ is the fundamental matrix solution of the original scattering problem such that $Y(x_0,x_0,z)= I$ (as in the previous section).
For each fixed value of~$z$ and $x_0$, \eqref{e:modZS} can be integrated numerically up to $x=x_0+L$ using standard methods (e.g., such as Runge-Kutta) for some fixed vector-valued initial condition. 
Letting $\^y(z)$ be the numerical solution of the modified scattering problem at $x=x_0+L$ with base point $x=x_0$ and initial condition (IC) $\~y(x_0,z) = (0,1)^T$,
the problem of finding the Dirichlet eigenvalues then reduces to that of finding those values of $z$ for which the top element of $\^y(x_0+L,z)$ equals zero.
To this end, one can use Newton's method, by constructing the iteration
\be
    z_{n+1}=z_n - {\^y_1(z_n)}/{\^y_1'(z_n)}\,,
\label{e:Newton}
\ee
where the prime denotes differentiation with respect to~$z$ and $\^y_1(z)$ denotes the top entry of $\^y(z)$.
This method requires knowledge of the derivative of $\^y_1(z)$ with respect to $z$, which can be obtained by constructing the first variational system. 
Differentiating~\eqref{e:modZS} with respect to~$z$ yields the forced linear system of ODEs
\be
    \~y'_x
     = U(-iz\sigma_3+Q(x))U^{-1}\,\~y' - iU\sigma_3 U^{-1}\~y\,.
\label{e:modZS'}
\ee
The first variational system is the $4\times4$ system of ODEs obtained by combining
\eqref{e:modZS} and~\eqref{e:modZS'},
namely,
\begin{equation}
    P_x =
    \begin{pmatrix} U(-iz\sigma_3+Q(x))U^{-1} & 0 \\
        -iU\sigma_3 U^{-1} & U(-iz\sigma_3+Q(x))U^{-1} \end{pmatrix} \, P\,,
\label{e:Pode}
\end{equation}
where the dependent variable is the four-component vector $P(x,z) = (\~y,\~y')^T$.
The IC for~\eqref{e:modZS'} is the derivative with respect to $z$ of the IC for the modified scattering problem~\eqref{e:modZS}, which implies $\~y'(x_0,z) = (0,0)^T$.
Let $\^P(z)$ be the numerical solution of the system at $x=x_0+L$ with IC $P(x_0,z) =(0,1,0,0)^T$. 
Then \eqref{e:Newton} becomes
\vspace*{-1ex}
\begin{equation}
    z_{n+1} = z_n- \^P_1(z_n)/\^P_3(z_n),
\label{e:Newton2}
\end{equation}
where $\^P_1(z)$ and $\^P_3(z)$ denote respectively the first and third entries of~$\^P(z)$. 
Allowing Newton's method to iterate until a desired tolerance is reached will then produce a Dirichlet eigenvalue with numerical error equal to that of the method used to solve the first variational system plus the tolerance of Newton's method.

The last issue that needs to be discussed to make this approach practical is the choice of $z_0$, namely, the initial condition for the iteration~\eqref{e:Newton2}. 
Recall that Newton's method only converges quickly if the initial guess is close to the true zero, and may not converge at all if the initial guess is too far away. 
In addition, the problem admits an infinite number of Dirichlet eigenvalues, and we would like to find as many of them as quickly as possible.
A possible way to do this is to select a rectangular region in the complex plane to check for Dirichlet eigenvalues. 
Next, one partitions this region with $m_\re$ points along the real axis and $m_\im$ points along the imaginary axis, thereby forming an $m_\re\times m_\im$ grid.
For each point $z$ in this grid, one integrates numerically~\eqref{e:Pode}
and then uses a bisection method to locate those points (which will generically lie in between grid points) where the real part of $\^P_1(z) = \~y_{12}(x_0+L,z)$ is zero. 
Connecting these points, one obtains the contour curves $\Re[\~y_{12}(x_0+L,z)] = 0$.
Repeating this process to construct the contours $\Im[\~y_{12}(x_0+L,z)] = 0$, 
the intersection between these zero contours will yield an acceptable initial guess for Newton's method.
Two issues with this method, however, are that (i) improving the accuracy of the zero contours requires a denser grid, and (ii) grid points that do not have a zero contour in their immediate vicinity still need to be calculated. 
Considering every grid point requires the numerical solution of a system of ODEs, so it would be desirable to find a way to determine initial guesses for Newton's method that reduce the required number of grid points. 

\begin{figure}[t!]
\kern-\medskipamount
\centerline{\includegraphics[trim= 0 0 0 0,clip,width=\figwidth]{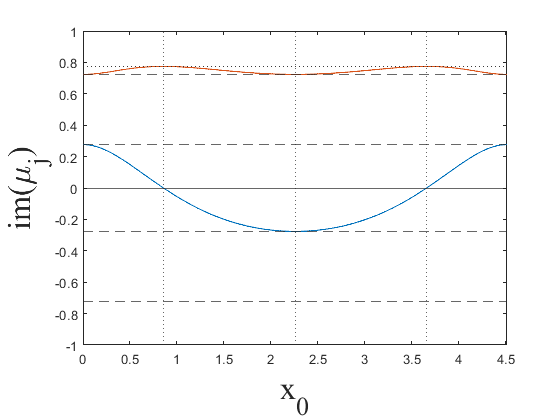}\kern-1em
  \includegraphics[trim= 0 0 0 0,clip,width=\figwidth]{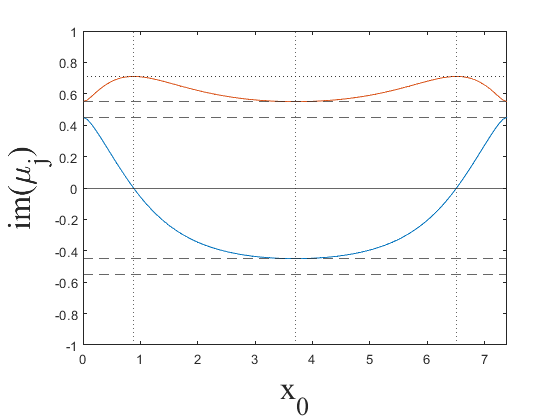}}
\kern-\bigskipamount
\caption{Numerically calculated movable Dirichlet eigenvalues along the imaginary axis as a function of the base point $x_o$ for the potential $q(x) = \dn(x,m)$ for $m=0.8$ (left) and $m=0.99$ (right).}
\label{f:Dirichletmotion}
\end{figure}

Figure~\ref{f:genuszero}(left) shows the spectrum for the constant potential $q(x) = 1$, calculated using the methodology described above. 
Since there are no movable Dirichlet eigenvalues, as shown in~\eqref{ytildeconstant}, 
it would be pointless to show multiple plots with different base points. 
Note however that even the immovable Dirichlet eigenvalues can be time-dependent.  
Recall that the location of the Dirichlet eigenvalue at $z=\i A\cos\alpha$ in \eqref{ytildeconstant} changes with~$\alpha$. 
Moreover, the IC $q(x,0)=1$ corresponds to the exact solution $q(x,t)=\e^{2\i t}$. Thus, if one takes $\alpha=2t$, the phase shift in the constant potential is the same as that obtained from the time evolution of the exact solution. 
Indeed, Fig.~\ref{f:genuszero}(right) shows how the Dirichlet spectrum for the potential
$q(x) = \e^{\i\pi/3}$ differs from the original one, validating the results of the previous section. 
On the other hand, we emphasize that the results of section~\ref{s:time}
show that the location of the poles for the time-dependent Riemann-Hilbert problem
remains unchanged with time, thanks to the effect of the functions $e^\pm(t,z)$.
Figure~\ref{f:genusone} shows plots of the spectrum for the genus-1 potential $q(x)=$dn$(x,0.8)$ with a variety of different choices for~$x_0$. 
Note that all Dirichlet eigenvalues on the real axis lie on a periodic or antiperiodic eigenvalue and are also immovable.
Finally, Fig.~\ref{f:Dirichletmotion} shows the dependence of the numerically calculated movable Dirichlet eigenvalues along the imaginary axis as a function of the base point $x_o$ for the potential $q(x) = \dn(x,m)$ for $m=0.8$ (left) and $m=0.99$ (right).
The values agree with the values obtained from calculations presented in the previous section.
We reiterate, however, that the inverse spectral method presented in section~\ref{s:inverse} only requires knowledge of the Dirichlet eigenvalues with a single base point.

\medskip
\let\em=\it

\makeatletter
\def\@biblabel#1{#1.}
\def\doibase{http://dx.doi.org/}
\def\reftitle#1{``#1''}
\def\booktitle#1{\textit{#1}}
\def\href#1#2{#2}
\makeatother
\small



\end{document}